\documentclass[11pt]{article}
\usepackage{geometry}
\usepackage[english]{babel}
\usepackage[utf8]{inputenc}
\usepackage{amsmath, amssymb, amsthm}
\usepackage[affil-it]{authblk}
\usepackage{natbib}
\usepackage{graphicx}
\usepackage{tabularx}
\usepackage{multirow} 
\usepackage{enumitem} 
\usepackage{subfigure} 
\usepackage{tablefootnote}
\usepackage{bigints}
\usepackage{xcolor, soul}

\usepackage[
pdftex,
colorlinks = true,
linkcolor = blue,
citecolor = blue,
filecolor = red,
urlcolor = blue
]{hyperref}
\newtheorem{theorem}{Theorem}
\newtheorem{lemma}{Lemma}
\newtheorem{definition}{Definition}
\newtheorem{corollary}{Corollary}
\graphicspath{ {./figures} }
\title{A Real Time Monitoring Approach for Bivariate Event Data}
\author[1,4]{Inez Maria Zwetsloot}
\author[2,*]{Tahir Mahmood}
\author[3]{Funmilola Mary Taiwo}
\author[1]{Zezhong Wang}
\affil[1]{Department of Advanced Design and Systems Engineering, City University of Hong Kong, Tat Chee Avenue, Kowloon, Hong Kong}
\affil[2]{Department of Technology, School of Science and Technology, The Open University of Hong Kong, 30 Good Shepherd Street, Ho Man Tin, Kowloon, Hong Kong}
\affil[3]{Department of Statistics, University of Manitoba, Winnipeg, MB R3T 2N2 Canada}
\affil[4]{School of Data Science, City University of Hong Kong, Tat Chee Avenue, Kowloon, Hong Kong}
\affil[*]{Corresponding author: Tahir Mahmood, tmahmood@ouhk.edu.hk}
\begin{document}
\maketitle

\begin{abstract}
Early detection of changes in the frequency of events is an important task, in, for example, disease surveillance, monitoring of high-quality processes, reliability monitoring and public health. In this article, we focus on detecting changes in multivariate event data, by monitoring the time-between-events (TBE). Existing multivariate TBE charts are limited in the sense that, they only signal after an event occurred for each of the individual processes. This results in delays (i.e., long time to signal), especially if it is of interest to detect a change in one or a few of the processes. We propose a bivariate TBE (BTBE) chart which is able to signal in real time. We derive analytical expressions for the control limits and average time-to-signal performance, conduct a performance evaluation and compare our chart to an existing method. The findings showed that our method is a realistic approach to monitor bivariate time-between-event data, and has better detection ability than existing methods. A large benefit of our method is that it signals in real-time and that due to the analytical expressions no simulation is needed. The proposed method is implemented on a real-life dataset related to AIDS.  
\end{abstract}

\newpage

\section{Introduction}\label{Sec:introduction}
Many diseases (e.g., chronic diseases) can be treated better if abnormal behavior of the disease is detected early. Early event detection is thus a critically important research problem in medical research and health surveillance \citep{mahmood2019monitoring}. We focus on the important question of how we can identify people with irregular longitudinal patterns of disease behavior. Acquired Immune Deficiency Syndrome (AIDS) is a chronic disease, which is a result of Human Immunodeficiency Virus (HIV) infection. Transfusion related AIDS data are collected by the center for disease control \citep{hu2014cross,moreira2021nonparametric}. The data consist of patients who were infected with HIV through blood or blood-product transfusion. The data records include the age of patients (categorized in adults and children), infection time (months between blood transfusion date (01 April, 1978) to HIV infection), induction time (months from HIV infection to AIDS diagnosis) and total time (sum of infection and induction time). We are interested in developing an event-based monitoring method which can be used for early detection of irregularity in event data, like this AIDS example.

Besides the application of early detection of irregular disease behavior, the detection of changes in dynamic event data has many other applications. For example, other healthcare related examples are monitoring of times to blindness in the eyes \citep{huster1989modelling,li2012statistical}, response time to different treatmeants \citep{gross1981paired,lu1991inference}, or recurrence time after (cancer) treatment \citep{byar1980veterans,chiou2018semiparametric}. In manufacturing processes, these methods can be used to detect changes in the production times of different batchesand to monitor the failure time systems \citep{flury2018multivariate,nelson1982applied}. Moreover, these methods can also be used for syndromic disease surveillance to detect specific symptoms in order to have early detection of disease outbreaks \citep{sparks2019real}.

Two types of methods have been proposed in the literature for monitoring of event data. One group of methods is the \emph{monitoring of count data}. Count data can be obtained by counting the number of occurred events in pre-specified time periods. Brief reviews of the monitoring methods for count data can be found in \cite{saghir2015control,ali2016overview} and \cite{ mahmood2019models}. However, the count based approach is not a real-time approach as one needs to wait until the end of each time period, e.g. a day, before changes can be detected \citep{zwetsloot2019review}. In addition, the selection of aggregation window length is always somewhat arbitrary. The second group of methods are \emph{Time-Between-Events} (TBE) control charts. With a TBE control chart, we can monitor the length of time between events. 
For recent studies on univariate TBE control charts, the reader is referred to \citet{sparks2019real,sparks2020monitoring}. Methods for multivariate TBE data are categorized into two types; methods for (a) vector-based event data and (b) point-process data \citep{zwetsloot2020multivariate}. Vector-based event data occur when multivariate TBE data are observed one vector at a time. For example, consider manufactured items that pass various process steps and for each step a processing time is recorded. This forms a vector of failure times: we obtain one observation for each step (assuming items do not get reprocessed). In point-process data, one event may occur several times before another showed up. For example, consider manufactured items that can fail in several ways and are repaired after failure. One may observe a failure type A twice before observing failure type B or C.  Our case study of interest, detection of abnormal behavior in disease behavior using the AIDS data, is an example of bivariate vector-based event data. Therefore, subsequently, we will focus on bivariate vector-based data. 

As far as we are aware, all literature on \emph{Multivariate Time-Between-Events} (MTBE) control charts are designed for vector-based event data. The most well known method is developed by \cite{xie2011two}. The authors considered Gumbel's Bivariate Exponential (GBE) distributed data and proposed a vector-based Multivariate Exponentially Weighted Moving Average (MEWMA) control chart. We will provide more details on existing MBTE control chart literature in Section 5. Noteworthy is that \emph{all} existing MTBE control chart methods have a built-in detection delay, which requires that one event is available for each of the $p$ variables under consideration. Hence, changes can only be detected when each variable has an observed event, as these events happen asynchronously in time, the methods have a built-in delay until the vector of event data is completely observed. For example, existing methods can only signal when we have observations of both events (e.g., in AIDS data; infection time and total time). Obviously, delays are undesirable when we wish to detect changes in the process as quickly as possible. Furthermore, it is easy to see how an extension from bivariate to a multivariate process will result in even longer delays in forming the vectors used for monitoring.

Therefore, in this article, we propose a novel and effective new method for \emph{real-time bivariate event-based monitoring}, called the bivariate timbe between event (BTBE) control chart. This method is designed for multivariate event data of the vector-based type. Our proposed method has real-time detection power and does not have a built-in delay like the existing methods. For instance, when we are interested in monitoring a patient's events time (say, infection time ($X_1$) and total time ($X_2$)) due to transfusion of blood, time $X_1$ is observed first and after which time $X_2$ is observed. In the proposed BTBE chart, time $X_1$ is plotted first and thereafter time $X_2$ is plotted for the monitoring. By this exercise, there is no need to wait for the occurrence of $X_2$ to signal a change in $X_1$. Moreover, an additional advantage of the proposed method is that it provides exact information about the root cause behind an out-of-control signal. We derive analytical expressions for the control limits and the average time-to-signal \emph{(ATS)}. 

The remainder of this article is organized as follows. The proposed method is introduced in Section~\ref{Sec:proposed_method}. Analytical expression for the theoretical performance of our method are given in Section~\ref{performance_measure}. The performance of our proposed method under different distributional environments is discussed in Section~\ref{Sec:sim.study} and a comparison with an existing method is presented in Sections~\ref{Sec:comparison}. Implementation of the proposed method on the real-life scenario is discussed in Section~\ref{Sec:case_study} and finally, the article is summarized in Section~\ref{Sec:conclusions}. Moreover, mathematical proofs and other details are provided in the \emph{Appendix} A-C and the supplementary material.

\section{Proposed Method}\label{Sec:proposed_method}
In this section, we present our proposed BTBE chart to monitor bivariate vector-based event data. This method has the ability to signal changes as data comes in and unlike the existing methods it has therefore no need to wait until we observe a complete vector of events.

\subsection{Data model and details}
Consider $X=(X_1,X_2)$ as a vector of bivariate lifetimes, where $X_1$ indicates the time to an event in the first subprocess and $X_2$ indicates the time to an event in the second subprocess. We assume that $(X_1,X_2)$ are drawn from a bivariate continuous probability density function $f(x_1,x_2; \theta)$ where $\theta$ is the parameter vector. We will discuss some typical choices for event time distributions $f()$ in Section \ref{sec:phase1}. We denote the corresponding cumulative joint distribution function, the joint survival function and the partial survival functions by 
\begin{equation} \label{eq:FS}
	\begin{split}
		F(x_1,x_2) &= P[X_1 \le x_1, X_2 \le x_2] \\
		S(x_1,x_2) &= P[X_1>x_1, X_2>x_2] \\
		S_1 (x_1,x_2) &= \frac{\partial}{\partial x_1} S(x_1,x_2)\\
		S_2 (x_1,x_2) &= \frac{\partial}{\partial x_2} S(x_1,x_2) 
	\end{split}
\end{equation}
As $(X_1,X_2)$ denote \emph{event times}, one of the two is observed \emph{first}. In order to model the data in real-time we define order statistics: $X_{(1)}$ is the first observed event time and $X_{(2)}$ is the second observed event time:
\begin{align} \label{eq:x}
	X_{(1)} = \min (X_1,X_2), \;\;\;\;\; X_{(2)} = \max(X_1,X_2)
\end{align}
As example, consider these four artificial event vectors:
\begin{equation*}
	X  = \begin{bmatrix} x_1 \\ x_2 \end{bmatrix} 
	= \begin{bmatrix} 2 \\ 3 \end{bmatrix} 
	, \begin{bmatrix} 3 \\ 1 \end{bmatrix} 
	, \begin{bmatrix} 2 \\ 2 \end{bmatrix} 
	, \begin{bmatrix} 5 \\ 1 \end{bmatrix}. 
\end{equation*}
The superimposed process consist of the events as they are observed and is given by 
\begin{align} \label{eq:super}
	[ 2,3,1,3,2,1,5]. 
\end{align}	
Note that for event 1 we observed $x_1$ first, for events 2 and 4 we observed $x_2$ first, and for event 3 we observed the two events at the same time and hence only have one time-between-events value in the superimposed process (Equation \eqref{eq:super}). For our method we plot the events as soon as the events occur and can therefore provide real-time signals. We plot each event consecutively on a chart, Figure \ref{fig:ex_btbe} illustrates this. 

\begin{figure}[h]
	\centering 
	\includegraphics[width=0.7\textwidth]{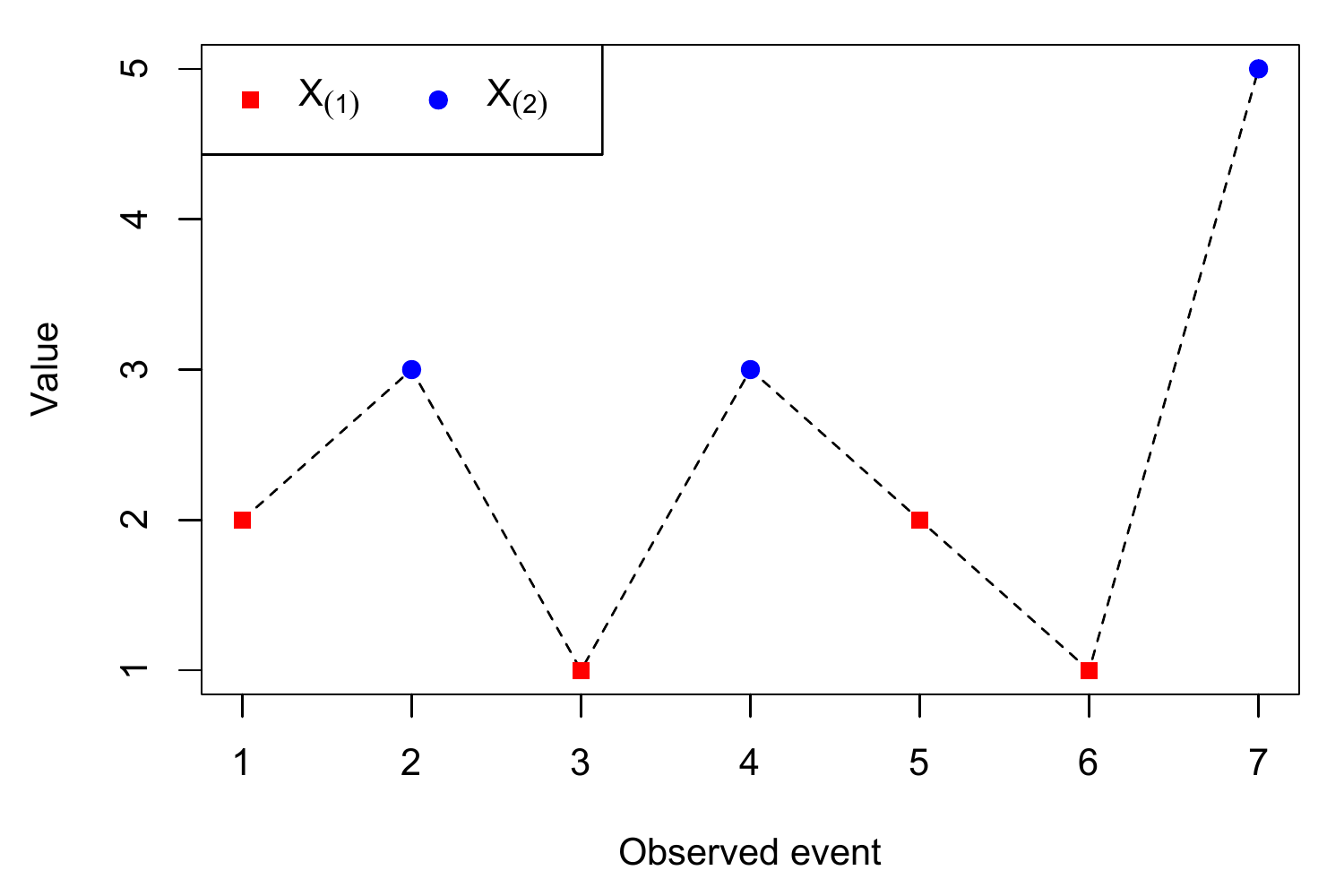}
	\caption{Illustration of superimposed data for four bivariate event vectors}
	\label{fig:ex_btbe}
\end{figure}
 
\subsection{Real-time distribution for Bivariate Event Data}
Control charts consist ofa plotting statistic and control limits. The plotting statistic (Q) equals the superimposed process as illustrated in Equation \eqref{eq:super}. Control limits are usually defined such that $F(LCL)=\alpha/2$ and $S(UCL)=\alpha/2$ for the lower (LCL) and upper (UCL) limits, respectively. Here $F()$ and $S()$ are the cumulative distribution and survival function for the univariate data stream under consideration. Furthermore, $\alpha$ is the false alarm rate and the chart signals when $Q<LCL$ or $Q>UCL$. 

In our case, we have a superimposed data stream, hence $F$ and $S$ will differ depending on whether the datapoint is observed first ($X_{(1)}$) or second $(X_{(2)})$ and whether the datapoint is from the first subprocess $(X_1)$ or the second subprocess $(X_2)$. We derive control limits for all these different situations. Therefor, our method will have dynamic control limits that differ for each plotted event. 

Next, to develop our method, we derive the cumulative distribution and survival functions of $X_{(1)}$ and $X_{(2)}$. Results are given in the Theorem 1 and 2, proofs can be found in \emph{Appendix} A.

\begin{theorem}\label{thm:F_x1}
	Assume $X=(X_1,X_2) \sim f(x_1,x_2)$ and let $X_{(1)}$ and $X_{(2)}$ be defined as in Equation~\eqref{eq:x}. Then the conditional cumulative distribution function of $X_{(1)}$ is defined as $F_{X_{(1)}}(x_{(1)})=P[X_{(1)}\leq x_{(1)}]/P[X_{(1)}<X_{(2)}]$ which is equal to
	\begin{align} \label{eq:Fx}
    F_{X_{(1)}}(x_{(1)}) = 
    \begin{cases}
		    	 \dfrac{\bigintsss_0^{x_{(1)}} S_1 (x_1,x_1) dx_1}{ \bigintsss_0^{\infty} S_1 (x_1,x_1) dx_1}   \;\;\; \text{if}\; X_1 < X_2 \\
				 \dfrac{\bigintsss_0^{x_{(1)}} S_2(x_2,x_2) dx_2 }{ \bigintsss_0^{\infty} S_2 (x_2,x_2) dx_2}  \;\;\; \text{if} \; X_1 > X_2  \\
				 \dfrac{\bigintsss_0^{x_{(1)}} f(x_1,x_1) dx_1}{ \bigintsss_0^{\infty} f(x_1,x_1) dx_1}      \;  \;\;\; \text{if} \; X_1 = X_2
	\end{cases}
\end{align}
Equivalently, the conditional survival function of $X_{(1)}$ is defined as 
$$S_{X_{(1)}}(x_{(1)})=P[X_{(1)}>x_{(1)}]/P[X_{(1)}<X_{(2)}]$$
which is equal to
\begin{align} \label{eq:Sx}
    S_{X_{(1)}}(x_{(1)}) = 
    \begin{cases}      
        \frac{ \bigintsss_{x_{(1)}}^{\infty} S_1(x_1,x_1)dx_1}{\bigintsss_0^{\infty} S_1 (x_1,x_1)dx_1} & \text{if}\; X_1 < X_2 \\ 
	 	\frac{ \bigintsss_{x_{(1)}}^{\infty} S_2 (x_2,x_2) dx_2}{ \bigintsss_{0}^{\infty} S_2 (x_2,x_2) dx_2} & \text{if}\; X_1 > X_2 \\
	 	\frac{ \bigintsss_{x_{(1)}}^{\infty} f(x_1,x_1)dx_1}{\bigintsss_{0}^{\infty} f (x_1,x_1)dx_1} & \text{if}\; X_1 =X_2 .
    \end{cases}
\end{align}
Where $S_1$ and $S_2$ are defined as in Equation \eqref{eq:FS}.
\end{theorem}


\par\noindent
For the second event time $X_{(2)}$ we derive the conditional cumulative distribution and survival functions. We condition on the realization of $X_{(1)}$ because the first event has been observed by the time we observe $X_{(2)}$.

\begin{theorem}\label{thm:F_x2}
	Assume $X=(X_1,X_2) \sim f(x_1,x_2)$ and let $X_{(1)}$ and $X_{(2)}$ be defined as in Equation~\eqref{eq:x}. Then the conditional cumulative distribution function $F_{X_{(2)}|X_{(1)}}(x_{(2)}|x_{(1)}) = P[X_{(2)} \leq x_{(2)}|X_{(1)}=x_{(1)}]$ is defined as 
	\begin{equation}\label{eq:Fx2}
	F_{X_{(2)}|X_{(1)}}(x_{(2)}|x_{(1)}) =
	\begin{cases}
	    1- \frac{S_1(x_1,x_2)}{S_1(x_1,x_1) } & \text{if}\; X_1 < X_2 \\
	    1- \frac{S_2(x_1,x_2)}{S_2(x_2,x_2)} & \text{if}\; X_1>X_2 \\
	    \end{cases}
	\end{equation}
Similarly, the conditional survival function $S_{X_{(2)}|X_{(1)}}(x_{(2)}|x_{(1)}) = P[X_{(2)} > x_{(2)}|X_{(1)}=x_{(1)}]$ is defined as
	\begin{equation} \label{eq:Sx2}
	S_{X_{(2)}|X_{(1)}}(x_{(2)}|x_{(1)}) =
	\begin{cases}
	  \frac{S_1(x_1,x_2)}{S_1(x_1,x_1) } & \text{if}\; X_1 < X_2 \\
	  \frac{S_2(x_1,x_2)}{S_2(x_2,x_2)} & \text{if}\; X_1>X_2 \\
	\end{cases}
	\end{equation}
\end{theorem}

\noindent
Note that the case $X_1=X_2$ is not included in Theorem \ref{thm:F_x2}, because by definition there is no second event time when both are observed at the same time.

\subsection{Proposed Monitoring Method for Bivariate Vector-based Event Data}\label{Sec:BTBE_chart}

We design two variants of our BTBE chart: a one-sided upper and a two-sided control chart. 

\subsubsection{The one-sided upper BTBE chart}\label{Sec:One-sided-BTBE}
A one-sided control chart is proposed for monitoring increases in the mean of bivariate event data. The chart consists of plotting the following two components:
\begin{align*}
	\begin{split}
		UCL^i_{(j)} \;\; \text{for} \;\; j=1,2 \;\; \text{and} \;\; i=1,2,3,....	 \\
		X^i_{(j)} \;\; \text{for} \;\; j=1,2 \;\; \text{and} \;\; i=1,2,3,....	
	\end{split}
\end{align*}
where $X^i_{(j)}$ is the superimposed data stream. To obtain the upper control limits we use Equation~\eqref{eq:Sx} in Theorem \ref{thm:F_x1} and Equation \eqref{eq:Sx2} in Theorem \ref{thm:F_x2}. We set $UCL^i_{(1)}$ such that $\alpha  = S_{X_{(1)}}(UCL^i_{(1)})$ and we set $UCL^{i}_{(2)}$ such that $\alpha = S_{X_{(2)}|X_{(1)}}(UCL^i_{(2)}|x^i_{(1)})$. 


\subsubsection{Two-sided BTBE chart} \label{Sec:Two-sided-BTBE}
A two-sided control chart is proposed for monitoring both increases and decreases in the mean of the bivariate event data. The chart consists of plotting the following three components on a chart:
\begin{align*}
	\begin{split}
		UCL^i_{(j)} \;\; \text{for} \;\; j=1,2 \;\; \text{and} \;\; i=1,2,3,....	 \\
		X^i_{(j)} \;\; \text{for} \;\; j=1,2 \;\; \text{and} \;\; i=1,2,3,....		\\
		LCL^i_{(j)} \;\; \text{for} \;\; j=1,2 \;\; \text{and} \;\; i=1,2,3,....	 
	\end{split}
\end{align*}
where $X^i_{(j)}$ is the superimposed data stream. To obtain the upper control limits we use Equation~\eqref{eq:Sx} in Theorem \ref{thm:F_x1} and Equation \eqref{eq:Sx2} in Theorem \ref{thm:F_x2}. We set $UCL^i_{(1)}$ such that $\alpha/2  = S_{X_{(1)}}(UCL^i_{(1)})$ and we set $UCL^{i}_{(2)}$ such that $\alpha/2 = S_{X_{(2)}|X_{(1)}}(UCL^i_{(2)}|x^i_{(1)})$. To obtain the lower control limits we use Equation~\eqref{eq:Fx} in Theorem \ref{thm:F_x1} and Equation \eqref{eq:Fx2} in Theorem \ref{thm:F_x2}. We set $LCL^i_{(1)}$ such that $\alpha/2  = F_{X_{(1)}}(LCL^i_{(1)})$ and we set $LCL^i_{(2)}$ such that $\alpha/2 = F_{X_{(2)}|X_{(1)}}(LCL^i_{(2)}|x^i_{(1)})$. 
\\
\par\noindent
For both the one-sided upper and two-sided control charts, our method will have dynamic control limits that differ depending on the data at hand.

\subsubsection{The false alarm rate ($\alpha$)}  \label{Sec:alpha}
Control limits are usually set by setting $\alpha = \frac{1}{ARL_0}$, where the in-control average run length ($ARL_0$) is set to a pre-defined value. However, when the data represents \emph{event times}, it is common to design the chart using the average time-to-signal (ATS) rather than the ARL \citep{zwetsloot2020multivariate}.  We can relate the ARL and ATS by $ATS = ARL*E[TBE]$ and as $ARL=1/\alpha$ it holds that:
\begin{align} \label{eq:alpha}
	\alpha = \frac{E[TBE]}{ATS_0}
\end{align}
where $E[TBE]$ is the expected time between two events and $ATS_0$ is the in-control desired ATS value. The expected time between two events ($E[TBE]$) will be derived in Section~\ref{performance_measure}.

\subsection{Implementation of our proposed BTBE chart}\label{Sec:BTBE_charts}
Control charts are generally implemented in two phases: phase I for determining and estimating a distribution model for the data and phase II for the prospective monitoring \citep{montgomery2017design}. 

\subsubsection{Phase I} \label{sec:phase1}
In Phase I, a stable and in-control dataset should be gathered. Assume we have $X^i$ for $i=1,2,...,n$ observations with $X^i=(X^i_1,X^i_2)$ - a bivariate vector of events times - collected asynchronously in time. Next a distribution function $f(x;\theta)$ should be selected and fitted. A variety of possible models $f()$ have been proposed in the literature \citep{kotz2004continuous}, they have a long history and go back to the 1960’s. One important consideration to keep in mind during the selection process of an appropriate model is the failure mechanism of the underlying process. We introduce the three most well-known models:

\begin{itemize}
    \item The \textbf{Gumbel's Bivariate Exponential (GBE)} distribution is based on an random external stress factor and was introduced by \citet{gumbel1960} and further developed by \citet{hougaard1986}.
    \item The \textbf{Marshall Olkin Bivariate Exponential (MOBE)} distribution \citep{marshall1967multivariate} is based on the assumption that the life-time of the main system does not depend on the failure time of the two-components, but it is affected by another common external factor. In this model, random shocks to the system appear as a homogeneous Poisson process.
    \item The \textbf{Marshal Olkin Bivariate Weibull (MOBW)} distribution was also introduced by \cite{marshall1967multivariate}. It is based on the MOBE model but allows for a more flexible shock process which is modeled using a non-homogeneous Poisson process.
\end{itemize}

The model should be carefully selected regarding the failure mechanism of the underlying process and/or by using some test to fit the most appropriate model. Simultaneously the model parameters should be estimated. More details regarding each of these models, including the maximum likelihood estimation equations, can be found in \emph{Appendix} B. The practitioner is of course free to select any other model, our method works for any distribution selected. 

Next, the control limits for our proposed BTBE monitoring method should be derived as explained in Sections \ref{Sec:One-sided-BTBE} and \ref{Sec:Two-sided-BTBE}. For the above discussed bivariate life-time models, we provide the control limits in Table \ref{tab:cls} and the derivations for these limits are provided in the supplementary material.

Note that for the GBE and MOBE models we design an one-sided upper control chart. This is due to a philosophical issue that the exponential distribution has most of its probability mass at the origin. A lower control limit would be used to signal observations that come from the most likely area of the exponential distribution. Therefore, we signal data that are most likely to occur under the in-control scenario and we recommend using only an upper-sided chart for exponential data. If decreases in the expected event time are of interest, we recommend the user to fit a MOBW model or some other models not based on the exponential distribution. As the MOBW model has most of its probability mass towards the expected value of the data (for $\eta>1$). Thus, signaling an observation below the LCL has more meaningful interpretation as those observations are naturally scarce in an in-control dataset. Therefore, we design a two-sided control chart for MOBW model.

Before we can set the control limits, we need to select an appropriate value for the $ATS_0$ and set $\alpha$ according to Equation \eqref{eq:alpha}. 

\begin{table}[h]
\centering
\small
\caption{Control limits for the BTBE chart for various life-time distributions}
\begin{tabular}{lllll}
\hline \hline
GBE &$X_{(1)}$&&$UCL_{(1)}$&$=-C(1,1)^{-\delta}\ln(\alpha)$\\
&&&&$\;\;\; where \; C(1,1)=\left( \frac{1}{\theta_1}^{1/\delta}+\frac{1}{\theta_2}^{1/\delta} \right)$\\
\cline{2-5}
&$X_{(2)}$&$\delta=1$&
$UCL_{(2)}$&$=
\begin{cases}
x_1 - \theta_2\ln(\alpha) \;\;\;\;\;\;\;\;x_1<x_2\\
x_2 - \theta_1\ln(\alpha) \;\;\;\;\;\;\;\;x_1>x_2
\end{cases}$ \\
\cline{3-5}
&&$0<\delta<1$
&$UCL_{(2)}$&$=
\begin{cases}
\left( \left(\frac{1-\delta}{\delta} \theta_2  W_0(G) \right)^{1/\delta} - \left( \frac{\theta_2 }{\theta_1} \; x_1 \right)^{1/\delta} \right)^{\delta} \;\;\;x_1<x_2\\
\left( \left(\frac{1-\delta}{\delta} \theta_1 W_0(G) \right)^{1/\delta} - \left( \frac{\theta_1}{\theta_2} \; x_2 \right)^{1/\delta} \right)^{\delta} \;\;\;x_1>x_2
\end{cases}$ \\
&&&&$\;\;\;W_0(.)$ is the Lambert function$^*$ \\
&&&&$\;\;\;G=\frac{\delta}{1-\delta} x_{(1)} C(1,1)^\delta \left( \alpha \exp \left(-x_{(1)} C(1,1)^\delta \right) \right)^{\frac{-\delta}{1-\delta}}$\\
\hline
MOBE &$X_{(1)}$&&$UCL_{(1)}$&$=-(\lambda_1+\lambda_2+\lambda_{12})^{-1}\ln(\alpha)$\\
\cline{2-5}
&$X_{(2)}$&&$UCL_{(2)}$&$=
\begin{cases}
x_1- (\lambda_2+\lambda_{12})^{-1}\ln(\alpha) \;\;\;\;\;\;\;\;x_1<x_2\\
x_2 - (\lambda_1+\lambda_{12})^{-1}\ln(\alpha) \;\;\;\;\;\;\;\;x_1>x_2
\end{cases}$ \\
\hline
MOBW &$X_{(1)}$&&$LCL_{(1)}$&$=\left(-(\lambda_1+\lambda_2+\lambda_{12})^{-1}\ln(1-\alpha/2)\right)^{1/\eta}$\\
&&&$UCL_{(1)}$&$=\left(-(\lambda_1+\lambda_2+\lambda_{12})^{-1}\ln(\alpha/2)\right)^{1/\eta}$\\
\cline{2-5}
&$X_{(2)}$&&$LCL_{(2)}$&$=
\begin{cases}
\left(x_1^{\eta} - (\lambda_2+\lambda_{12})^{-1}\ln(1-\alpha/2)\right)^{1/\eta} \;\;\;x_1<x_2\\
\left(x_2^{\eta} - (\lambda_1+\lambda_{12})^{-1}\ln(1-\alpha/2)\right)^{1/\eta} \;\;\;x_1>x_2
\end{cases}$ \\
&&&$UCL_{(2)}$&$=
\begin{cases}
\left(x_1^{\eta}- (\lambda_2+\lambda_{12})^{-1}\ln(\alpha/2)\right)^{1/\eta} \;\;\;\;\;\;\;\;x_1<x_2\\
\left(x_2^{\eta} - (\lambda_1+\lambda_{12})^{-1}\ln(\alpha/2)\right)^{1/\eta} \;\;\;\;\;\;\;\;x_1>x_2
\end{cases}$ \\
\hline \hline
\multicolumn{5}{l}{\footnotesize$^*$ The Lambert function $y=W_0(z)$ gives a solution for the equation $y \exp(y) = z$. }\\
\multicolumn{5}{l}{For details see \cite{lambert1758observationes}, \cite{euler1779serie}, and the supplementary material.}
\end{tabular}
\label{tab:cls}
\end{table}

\subsubsection{Phase II}
After obtaining a model, estimated parameters and control limits, the chart is ready to be run. In this section, we explain how to run the BTBE chart. 

For illustration purposes, we generate $5$ artificial data vector from a GBE model with parameters $\theta_1=20,\,\theta_2=15,\,\delta=0.5$. We assume that the in-control parameters are $\theta_1=5,\,\theta_2=15,\,\delta=0.5$. Using $ATS_0=200$ we obtain control limits based on the expressions in Table \ref{tab:cls}. Table \ref{tab:art.example} provides the details. 

The BTBE chart presented in Figure \ref{artf_data_plot}~(a) is a plot of the superimposed $X_{(1)}$ and $X_{(2)}$ process against their corresponding control limits. It is seen that the first signal occurred at the $6^{th}$ observed event and the corresponding digit `2' indicates that the signal is due to second event i.e. $X_{(2)}$. Similarly, three more signals occurred at the $15-17^{th}$ events, where $16^{th}$ signal was due to the second event while other two were due to the first event. 

Figure \ref{artf_data_plot}~(b) shows the data and control limits in xy-plane format, with $X_{1}$ on the x-axis and $X_{2}$ on the y-axis. The signals due to $X_{(1)}$ are shown with `\textcolor{blue}{$\times$}' and the signals due to $X_{(2)}$ with `\textcolor{red}{$\Box$}'.

\begin{table}[ht]
	\centering
	\small
	\begin{tabular}{c rrrr rr}
		\hline \hline
		Sample & $X_1$ & $X_2$ & $X_{(1)}$ & $X_{(2)}$ &$UCL_{(1)}$&	$UCL_{(2)}$\\
		\hline
1&	24&	10&	10&	24& 18.78	&25.64\\
2&	15&	22&	15&	22&	18.78&85.05\\
3&	36&	15&	15&	36&		18.78&31.68\\
4&	11&	8&	8&	11&		18.78&23.02\\
5&	17&	27&	17&	27&		18.78&89.89\\
6&	3&	2&	2&	3&		18.78&12.85\\
7&	2&	1&	1&	2&		18.78&9.73\\
8&	70&	49&	49&	70&	18.78&67.99\\
9&	28&	56&	28&	56&	18.78&113.20\\
10&	4&	2&	2&	4&		18.78&12.85\\
		\hline\hline 
	\end{tabular}
	\caption{An artificial example data set and corresponding control limits}
	\label{tab:art.example}
\end{table}

\begin{figure}[!htb]
	\centering
	\includegraphics[width=1\linewidth]{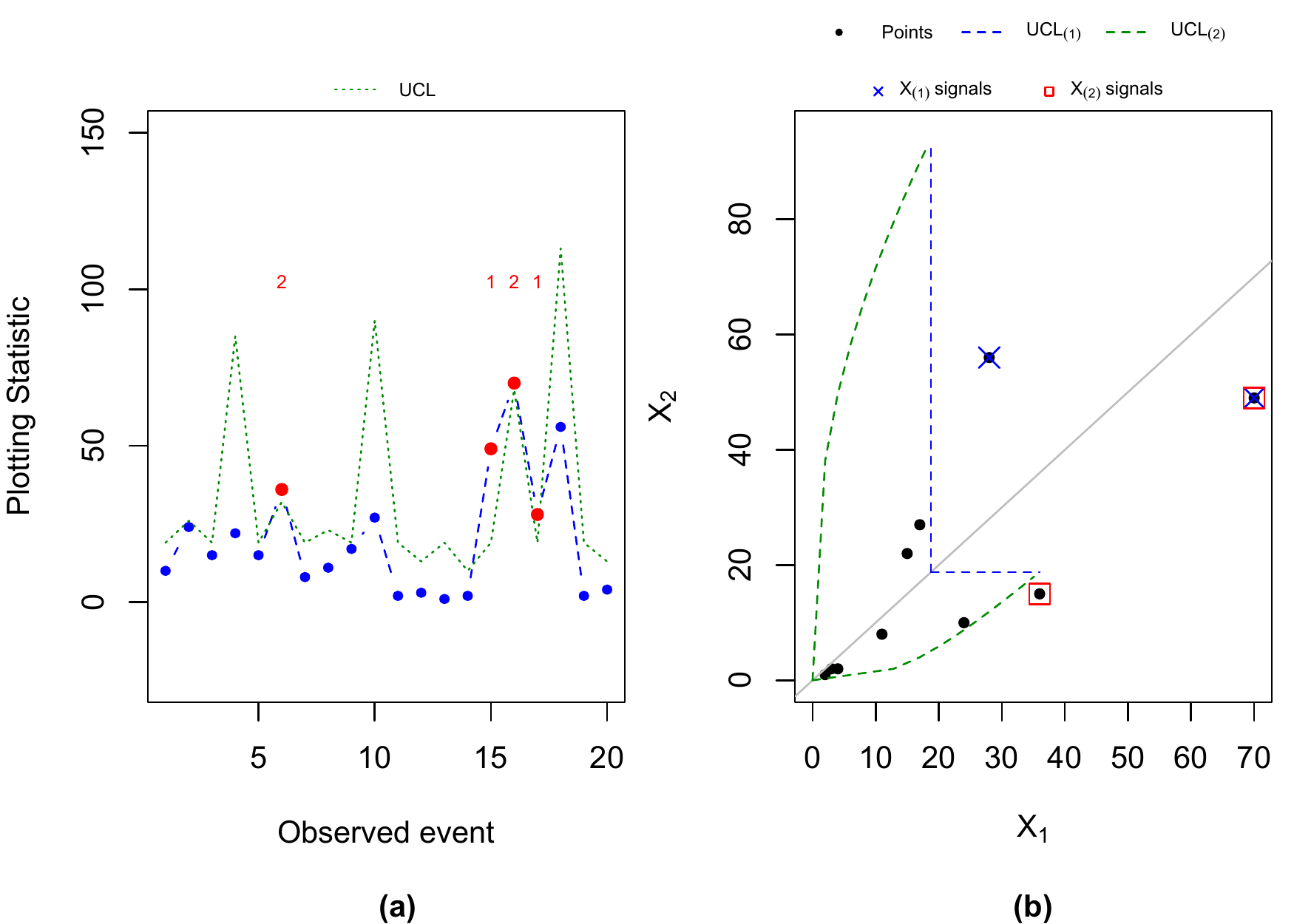}
	\caption{\label{artf_data_plot} Illustration of proposed BTBE chart on artificial data set in form of: (a) Superimposed format and (b) xy-plane format.}
\end{figure}

\section{Theoretical performance}\label{performance_measure}
Given a series of observed bivariate event data $\boldmath{X}^i=(X_1^i,X^i_2)$ for $i=1,2,...,$ the data model is assumed as follows: 
\begin{align} \label{eq:model}
	\begin{split}
		\boldmath{X}^i &\sim f(x|\theta) \;\;\;\; i \leq \tau \\
		\boldmath{X}^i &\sim f(x|\theta^*) \;\;\; i > \tau 
	\end{split}
\end{align}
Given this model, the objective of a control chart is to signal as quickly as possible after $\tau$. Usually, a control chart's performance is quantified using the Average Run Length (ARL) defined as the average number of points after $\tau$, until an out-of-control point is signaled. However, \citet{zwetsloot2020multivariate} highlighted that it is more appropriate to quantify the performance of methods for event data, based on time-to-signal performance metric because the time between plotted statistics varies when monitoring event data. Therefore, we use the average time-to-signal (ATS) as the performance metric, which is defined as the average time elapsed after $\tau$, until an out-of-control point is signaled. ATS can be computed as $ATS = ARL * E[TBE]$, where
\begin{align}  \label{eq:ETBE}
	E[TBE]  = P[X_1=X_2]E[X_{(1)}|X_1= X_2] + 0.5P[X_1\neq X_2] E[X_{(2)}|X_1\neq X_2] 
\end{align}
With probability $P[X_1=X_2]$ we can observe only one event with $E[X_{(1)}|X_1=X_2]$ being the expected time to this event. With probability $P[X_1 \neq X_2]$, we can either observe $X_{(1)}$ or we observe $X_{(2)}$ since both occur with probability $0.5$ and they are multiplied by the event time for $X_{(1)}$ and $X_{(2)}$, respectively. Thus, $$P[X_1\neq X_2] E[X_{(2)}|X_1\neq X_2] = P[X_1\neq X_2] E[X_{(1)}|X_1\neq X_2] + P[X_1\neq X_2] E[X_{(2)}-X_{(1)}|X_1\neq X_2].$$  

Since we used a superimposed data stream, not all data points are independent hence, the computation of $ARL$ and consequently, the $ATS$ is subtle. Taking this into account, we derive the analytical expression for the $ATS$ in Theorem \ref{thm:ats}. Proof of Theorem \ref{thm:ats} can be found in \emph{Appendix} A.

\begin{theorem}\label{thm:ats}
	Assume $X=(X_1,X_2) \sim f(x|\theta^*)$ as in model \ref{eq:model}. Let $UCL_{(1)}$ and $UCL_{(2)}$ be defined as in sections \ref{Sec:One-sided-BTBE} and \ref{Sec:Two-sided-BTBE} and obtained from an in-control data set $X\sim f(x|\theta)$. The proposed BTBE chart will have an ATS performance equal to: 
	\begin{align}\label{eq:an-ats}
		ATS = E_*[TBE]\frac{1 + P_*[NS_1,\neq ]}{P_*[S_1]+P_*[NS_1,S_2, \neq ]} 
	\end{align}
	where $P_*$ and $E_*$ denote the probability and expectation over $f(x|\theta^*$), respectively and $\neq$ is the abbreviation for $X_{1}\neq X_{2}$. When evaluating an one-sided chart, the probabilities are defined as:
	\begin{align} \label{eq:prob1}
		\begin{split}
		P_*[NS_1,\neq ] &= P_*[X_{(1)}\leq UCL_{(1)}, X_1 \neq X_2] \\
		P_*[S_1] &= P_*[X_{(1)}> UCL_{(1)}] \\
		P_*[NS_1,S_2,\neq ] &=  P_*[X_{(1)} \leq UCL_{(1)}, X_{(2)}>UCL_{(2)},X_1\neq X_2] 
		\end{split}
	\end{align}
	and for a two-sided chart they are defined as
	\begin{align} \label{eq:prob2}
		\begin{split}
		P_*[NS_1,\neq ] =& P_*[LCL_{(1)} < X_{(1)}\leq UCL_{(1)}, X_1 \neq X_2] \\
		P_*[S_1] =& P_*[X_{(1)} \leq LCL_{(1)}] + P_*[X_{(1)}> UCL_{(1)}] \\
		P_*[NS_1,S_2,\neq ] =&  P_*[LCL_{(1)} <  X_{(1)} \leq UCL_{(1)}, X_{(2)}\leq LCL_{(2)} ,X_1\neq X_2] \\
										 &+ P_*[LCL_{(1)} <  X_{(1)} \leq UCL_{(1)}, X_{(2)}>UCL_{(2)},X_1\neq X_2] 
		\end{split}
	\end{align}
\end{theorem}
\noindent Note that when $P[X_1=X_2]=0$ (as is the case in the GBE model), the expression for the $ATS$ in Equation~\eqref{eq:an-ats} simplifies to 
\begin{align}
	ATS = E_*[TBE]\frac{2 - P_*[S_1]}{P_*[S_1]+P_*[NS_1,S_2]}.
\end{align}
This is so because if  $P_*[X_1=X_2]=0$, it follows that $$P_*[NS_1,X_1\neq X_2]=P_*[NS_1]= 1- P_*[S_1].$$ Next, using Theorem \ref{thm:ats}, we derived analytical expressions for the $ATS$ performance when our data follows the three selected life-time distributions. Detailed derivations to obtain the results in Corollary \ref{cor:gbe}-\ref{cor:mobw} can be found in the supplementary material.

\begin{corollary}\label{cor:gbe}
	Assume $X \sim f(x; \theta)$ as in model \ref{eq:model}, where $f()$ is the GBE distribution function with $\theta=(\theta_1,\theta_2,\delta)$ for in-control data, and $\theta=(\theta^*_1,\theta^*_2,\delta^*)$ for out-of-control data. When $\delta=1$, the ATS performance of our one-sided proposed method, is equal to 
	\begin{align}
		ATS =  
		\frac{2 - \alpha^{\frac{C^*_{11}}{C_{11}}}}
		{1- \left(1- \alpha^{\frac{C^*_{11}}{C_{11}}}\right) 
			\left( \frac{\theta_1^{*-1}}{C^*_{11}} 
			\left( 1- \alpha^{\frac{\theta_2}{\theta_2^*}}\right)	+\frac{\theta_2^{*-1}}{C^*_{11}}      
			\left( 1- \alpha^{\frac{\theta_1}{\theta_1^*}} \right) \right)      
		}E_*[TBE]
	\end{align}	
	where 
	\begin{align*}
		E_*[TBE] = 0.5(\theta_1^*+\theta_2^*-C(1,1)^{-\delta})
	\end{align*}
\end{corollary}
\noindent Note that for $\delta<1$ the analytical expression for the $ATS$ is too difficult to derive, thus we use simulation to obtain the ATS.

\begin{corollary}\label{cor:mobe}
	Assume $X \sim f(x; \theta)$ as in model \ref{eq:model}, where $f()$ is the MOBE distribution function, with $\theta = (\lambda_1,\lambda_2,\lambda_{12})$ for in-control data, $\theta = (\lambda^*_1,\lambda^*_2,\lambda^*_{12})$ for out-of-control data, and $\Lambda=\lambda_1+\lambda+\lambda_{12}$. The ATS performance of our one-sided proposed method, is equal to 
	\begin{align}
		ATS = \frac{1+\frac{\lambda^{*}_{1} + \lambda^{*}_{2}}{\Lambda^*}\bigg(1-\alpha^{\frac{\Lambda^{*}}{\Lambda}}\bigg)}{\alpha^{\frac{\Lambda^{*}}{\Lambda}}+\bigg(1-\alpha^{\frac{\Lambda^{*}}{\Lambda}}\bigg)\bigg(\frac{\lambda^{*}_{1}}{\Lambda^{*}}\alpha^{\frac{\lambda^{*}_{2} + \lambda^{*}_{12}}{\lambda_{2} + \lambda_{12}}}+ \frac{\lambda^{*}_{2}}{\Lambda^{*}}\alpha^{\frac{\lambda^{*}_{1} + \lambda^{*}_{12}}{\lambda_{1} + \lambda_{12}}}\bigg)}E_*[TBE] 
\end{align}
	where 
\begin{align*}
	E_*[TBE] = 0.5\left(\frac{\lambda^*_1+\lambda^*_{2}}{{\Lambda^*}^{2}}+\frac{\lambda^*_{2}}{\Lambda^*(\lambda^*_{1}+\lambda^*_{12})}+\frac{\lambda^*_{1}}{\Lambda^*(\lambda^*_{2}+\lambda^*_{12})}\right)+\frac{\lambda^*_{12}}{{\Lambda^*}^{2}}.
\end{align*}
\end{corollary}

\begin{corollary} \label{cor:mobw}
	Assume $X \sim f(x; \theta)$ as in model \ref{eq:model}, where $f()$ is the MOBW distribution function, with $\theta = (\lambda_1,\lambda_2,\lambda_{12},\eta)$ for in-control data, $\theta = (\lambda^*_1,\lambda^*_2,\lambda^*_{12},\eta)$ for out-of-control data, and $\Lambda=\lambda_1+\lambda+\lambda_{12}$. The ATS performance of our two-sided proposed method, is equal to 
	{\scriptsize
	\begin{align} 
		ATS =
		\frac{ \left[1 + \frac{\lambda^{*}_{1}+\lambda^{*}_{2}}{\Lambda^{*}} \left( 1 - \alpha^* \right) \right]E_*[TBE]}{\alpha^*+ (1-\alpha^*) \left[ \frac{\lambda_1^*}{\Lambda^*}\left(1-(1-\alpha)^\frac{\lambda_{2}^{*}+\lambda_{12}^{*}}{\lambda_{2}+\lambda_{12}} + \alpha^\frac{\lambda_{2}^{*}+\lambda_{12}^{*}}{\lambda_{2}+\lambda_{12}} \right) + \frac{\lambda_2^*}{\Lambda^*}\left(1-(1-\alpha)^\frac{\lambda_{1}^{*}+\lambda_{12}^{*}}{\lambda_{1}+\lambda_{12}} + \alpha^\frac{\lambda_{1}^{*}+\lambda_{12}^{*}}{\lambda_{1}+\lambda_{12}} \right)\right]} 
	\end{align}}\par\noindent
where $\alpha^*_L = 1- (1-\alpha)^{\frac{\Lambda^{*}}{\Lambda}}$, $\alpha^*_U = \alpha^{\frac{\Lambda^{*}}{\Lambda}}$, and $\alpha^* = \alpha^*_L + \alpha^*_U $ and
{\scriptsize
	\begin{align*}
	E_*[TBE] = 0.5\Gamma\left(1+\frac{1}{\eta}\right) \left(\frac{1}{(\lambda^*_{2}+\lambda^*_{12})^{1/\eta}}-\frac{\lambda^*_{2}+\lambda^*_{12}}{{\Lambda^*}^{1+1/\eta}}+\frac{1}{(\lambda^*_{1}+\lambda^*_{12})^{1/\eta}}-\frac{\lambda^*_{1}+\lambda^*_{12}}{{\Lambda^*}^{1+1/\eta}}+2\frac{\lambda^*_{12}}{{\Lambda^*}^{1+1/\eta}}\right)
\end{align*}\par}
\end{corollary}

\section{Performance of proposed method}\label{Sec:sim.study}

In this section, we evaluate the performance of the proposed BTBE chart. First, we describe our experiments followed by a performance evaluation.

\subsection{Synthetic data experiments}
We consider data from either of the three life-time distributions introduced in Section \ref{Sec:BTBE_charts}.  We select four in-control models to evaluate the following scenarios. 
\begin{itemize}
	\item Scenario 1: equal expectations for both event times; $E[X_1] = E[X_2] = 5$, with $X_1$ and $X_2$ modeled to be independent.
	\item Scenario 2: equal expectations for both event times; $E[X_1] = E[X_2] = 5$, with $X_1$ and $X_2$ modeled to be dependent.
	\item Scenario 3: unequal expectations for the event times; $E[X_1] = 5$ and $E[X_2] = 15$, with $X_1$ and $X_2$ modeled to be independent.
	\item Scenario 4: unequal expectations for the event times; $E[X_1] = 5$ and $E[X_2] = 15$, with $X_1$ and $X_2$ modeled to be dependent.
\end{itemize}

For the performance assessment of the BTBE chart, four different types of shifts are considered for each of the in-control scenarios:
\begin{itemize}
    \item Shift type I1: Increase in only $E[X_1]$ by 50 and 100 percent. 
    \item Shift type I2: Increase in both $E[X_1]$ and $E[X_2]$ by 50 and 100 percent.
    \item Shift type D1: Decrease shift in $E[X_1]$ to 50 percent. 
    \item Shift type D2: Decrease shift in both $E[X_1]$ and $E[X_2]$ to 50 percent. 
\end{itemize}

\par\noindent
For each of the three distributions, we have selected parameters ($\lambda_1, \lambda_2,\lambda_{12}$ for MOBE, $\lambda_1, \lambda_2,\lambda_{12},\eta$ for MOBW and $\theta_1,\theta_2,\delta$ for GBE) such that the above mentioned expectations are met under in-control and out-of-control scenarios. Table \ref{tab:par} in \emph{Appendix} B provides the parameter values. 

Note that dependence has a careful interpretation when the data are MOBE or MOBW distributed. So dependence is related to the probability that $X_1=X_2$. Therefore to model independence in scenarios 1 and 3 we set $P[X_1=X_2] = 0  $ and to model dependence in scenarios 2 and 4 we set $P[X_1=X_2]= 0.1$. 

\subsection{Performance evaluation of BTBE chart}
Table \ref{tab:ats} shows the ATS performance of our proposed BTBE chart for the selected experiments and the three distributions. Note that we implement the one-sided upper chart for the GBE and MOBE distributed data and a two-sided chart for the MOBW distributed data. Results in Table \ref{tab:ats} where obtained using Corollary \ref{cor:gbe}-\ref{cor:mobw}, except for the GBE results in scenarios 2 and 4 (the dependent scenarios ($\delta<1$)). All code for replicating the results can be found on \url{https://github.com/tmahmood5/Codes-BTBE-Monitoring-Method}.  

\begin{table}[ht]
	\centering
	\small
	\caption{ATS values for our BTBE chart under various data distributions and scenarios}
	\begin{tabular}{llll rrr}
		\hline \hline
		&&&&\multicolumn{3}{c}{ATS }   \\
		\cline{5-7}
		In-control scenario&Shift type&$E[X_1]$ &$E[X_2]$ &GBE&MOBE&MOBW \\
		\hline
		1. Equal expectations,  &IC&5&5& 200.0 & 200.0 & 200.0\\
		independence&OC-I1&7.5&5& 110.5 & 110.5&67.0 \\
		&OC-I1&10&5& 79.4 & 79.4& 35.9\\
		&OC-I2&7.5&7.5& 79.7 & 79.7& 40.0 \\
		&OC-I2&10&10& 54.8 & 54.8& 21.4\\
		&OC-D1&2.5&5&*&* &133.6\\
		&OC-D2&2.5&2.5&*& *&50.6\\
		\hline
		2. Equal expectations,  &IC&5&5& 199.2 & 200.0 & 200.0\\
		dependence&OC-I1&7.5&5& 115.4& 110.1& 66.9\\
		&OC-I1&10&5&79.9 &78.6 & 35.4 \\
		&OC-I2&7.5&7.5&91.5 &79.8 &40.7 \\
		&OC-I2&10&10&63.4 & 54.9& 21.9\\
		&OC-D1&2.5&15&*&* &136.0\\
		&OC-D2&2.5&7.5&*& *&50.6\\
		\hline
		3. Unequal expectations,  &IC &5&15&200.0& 200.0&200.0  \\
		independence&OC-I1&7.5&15& 110.7&110.7 &71.5 \\
		&OC-I1&10&15&78.4 & 78.4& 37.3\\
		&OC-I2&7.5&22.5& 103.1 &103.1 &63.4\\
		&OC-I2&10&30&80.6 & 80.6& 40.5\\
		&OC-D1&2.5&15&*&* &138.0\\
		&OC-D2&2.5&7.5&*& *&51.5\\
		\hline
		4. Unequal expectations,  &IC &5&15&192.8 &200.0 & 200.0 \\
		dependence&OC-I1&7.5&15&108.8 & 111.7&73.8 \\
		&OC-I1&10&15& 73.7 & 79.1& 38.4\\
		&OC-I2&7.5&22.5& 109.5 & 103.2& 63.9 \\
		&OC-I2&10&30& 83.2 & 80.7& 41.1\\
		&OC-D1&2.5&15&*&* &139.3\\
		&OC-D2&2.5&7.5&*& *&51.5\\
		\hline\hline 
	\end{tabular}
	\label{tab:ats}
\end{table}
From Table \ref{tab:ats} we conclude that our method is able to signal shifts for all three data distributions. The method is a little bit slower when the data are dependent (scenarios 2 and 4) compared to similar independent scenarios (1 and 3), however the difference is small. Next, we compare the results of our BTBE method for paired observations with equal and unequal expectations (scenarios 1 and 2 versus scenarios 3 and 4). $ATS$ values for scenarios 3 and 4 are larger compared to scenarios 1 and 2, this is a direct result from the expectation of $X_2$ being larger in scenarios 3 and 4. 

The two-sided chart applied with the MOBW data has good performance for detecting decreases in both variables. The charts are all designed to have an in-control $ATS$ of 200, the results obtained using the analytical expression give exact results. The simulation results for GBE with $\delta<1$ (scenarios 2 and 4) shows a little bit of simulation error. We note that the $ATS$ values for GBE and MOBE under the independent scenarios (1 and 3) are equal, because both distributions are equal when the data are independent. 

\section{Comparative Analysis}\label{Sec:comparison}
In this section, we compare our proposed BTBE chart. Possible comparative methods to select from are the initial MBTE chart by \citet{xie2011two}. Or more recently, the method by \cite{xie2021multivariate} who extended the idea of \cite{xie2011two} by using the multivariate cumulative sum (CUSUM) control chart. \cite{koutras2017new} and \cite{triantafyllou2020distribution} used control charts based on the order statistic to monitor the bivariate vector-based data. A two-level multivariate Bayesian control chart based on the Marshall-Olkin bivariate exponential (MOBE) distributed data was proposed by \cite{duan2020two}. For bivariate vector-based event data, copula based MEWMA, multivariate double EWMA and multivariate CUSUM charts were proposed by \cite{kuvattana2015comparative}, \cite{sasiwannapong2019efficiency}, and \cite{sukparungsee2021effects}, and the Hotelling’s $T^2$ chart based on the different type of copulas was discussed by \cite{sukparungsee2018bivariate}. For the multivariate vector-based event data, copula based MCUSUM chart was proposed by \cite{sukparungsee2017multivariate} and the MEWMA charts based on transformed exponential data and asymmetric gamma distributions were discussed by \cite{khan2018multivariate} and \cite{flury2018multivariate}, respectively. 

\begin{table}[htbp]
	\caption{Limit $h$ for the MEWMA chart against a fixed $ATS_0=200$}
	\centering
	\begin{tabular}{rcccc}
		\hline \hline
		&\multicolumn{4}{c}{Scenario}\\
		\cline{2-5}
		&1 & 2  &3&4\\ 
		\hline
		$\lambda=0.1$ & 3.60 & 3.87 & 2.09 & 2.12 \\
		$\lambda=1$ &   9.51 & 11.40& 5.33 & 5.86 \\
		\hline\hline
	\end{tabular}
	\label{tab:MEWMA_limits}
\end{table}

For our comparison we select the method by \citet{xie2011two} as this is the most well-known and widely studied method. Their multivariate EWMA (MEWMA) chart is designed for GBE distributed data only. The MEWMA statistic is defined as: 
\begin{equation*}
		z_i =r(X_i-\mu_X)+(1-r)z_{i-1}
\end{equation*}
and the charting statistics is equal to 
\begin{equation*}
	   E_i=z^T_i \Sigma^{-1}_{Z_i} z_i 
\end{equation*}
where $r \in (0,1]$ is the EWMA smoothing parameter, $\Sigma_{Z_i}=r\Sigma_X/2-r$ and, $\mu_X$ and $\Sigma_X$ are the mean vector and covariance matrix respectively (see \citet{xie2011two} for more details). The MEWMA chart signals when $E_i >h$ and  this chart converts to a Hotelling's $T^2$ chart when $r=1$. We run the MEWMA chart with $r=0.1$ and $r=1$. Note that our chart is essentially a Shewhart-type chart and hence, it is most fair to compare our method with the MEWMA chart based on $r=1$.  

\begin{table}[ht]
	\caption{Comparison of ATS performance for our BTBE chart and the MEWMA chart}
	\centering
	\small
	\begin{tabular}{lllll rrr}
		\hline \hline
		&&&&&\multicolumn{3}{c}{ATS }   \\
		\cline{6-8} \\
		&&&&&BTBE&MEWMA&MEWMA \\
		In-control scenario&Shift type&$\theta_1$ &$\theta_2$&$\delta$ &&$r=0.1$&$r=1$ \\
		\hline
		1. Equal expectations, &IC&5&5&1& 200.0 & 200.3 & 199.6 \\
		  independence &OC-I1&7.5&5&1& 110.5 & 102.9 & 106.8 \\
		  &OC-I1&10&5&1& 79.4 & 73.4 & 75.6 \\
		  &OC-I1&20&5&1& 56.2 & 57.9 & 55.0 \\
		  &OC-I2&7.5&7.5&1& 79.7 & 81.1 & 79.8 \\
		  &OC-I2&10&10&1& 54.8 & 60.1 & 56.7 \\
		  &OC-I2&20&20&1& 40.5 & 54.2 & 48.7 \\
		\hline
		2. Equal expectations,  &IC&5&5&0.5& 199.2 & 200.5& 200.3 \\
		  dependence&OC-I1&7.5&5&0.5& 115.4 & 85.3& 103.6\\
		  &OC-I1&10&5&0.5& 79.9 & 59.6& 70.9\\
		  &OC-I1&20&5&0.5& 49.2 & 50.0& 50.7\\
		  &OC-I2&7.5&7.5&0.5& 91.5 & 86.8&  87.0\\
		  &OC-I2&10&10&0.5& 63.4 & 64.9& 63.0 \\
		  &OC-I2&20&20&0.5& 43.6 & 58.0& 54.2\\
		  \hline
		3. Unequal expectations,  &IC &5&15&1& 200.0 & 200.3& 200.0  \\
		  independence&OC-I1&7.5&15&1& 110.7 & 127.2&111.6 \\
		  &OC-I1&10&15&1& 78.4 & 90.0&79.8 \\
		  &OC-I1&20&15&1& 51.2 & 62.0&55.6 \\
		  &OC-I2&7.5&22.5&1& 103.1 & 126.6&108.0  \\
		  &OC-I2&10&30&1& 80.6 & 102.0&88.2 \\
		  &OC-I2&20&60&1& 72.4 & 103.6&90.6 \\
		  \hline
		4. Unequal expectations,  &IC &5&15&0.5& 192.8 & 199.6 &199.5  \\
		  dependence&OC-I1&7.5&15&0.5& 108.8 & 105.7 & 104.5\\
		  &OC-I1&10&15&0.5& 73.7 & 70.7 & 70.6 \\
		  &OC-I1&20&15&0.5& 44.0 & 48.5 & 46.3\\
		  &OC-I2&7.5&22.5&0.5& 109.5 & 139.6 & 119.3 \\
		  &OC-I2&10&30&0.5& 83.2 & 116.2 & 100.8\\
		  &OC-I2&20&60&0.5& 61.2 & 116.8 & 105.6\\
		\hline\hline 
	\end{tabular}
	\label{tab:compare}
\end{table}
We design our method and the MEWMA chart with an in-control overall $ATS_0=200$. The control limits for our BTBE method are taken from Table \ref{tab:cls}. For the MEWMA chart, we obtain the control limits using simulation, the limits are reported in Table \ref{tab:MEWMA_limits}. Note that these limits are different compared to the limits reported by \citet{xie2011two} because we use the $ATS$ as performance measure and they used the $ARL$. 

Table \ref{tab:compare} gives the $ATS$ values for our chart and the MEWMA chart when the data are drawn from a GBE distribution with scenarios similar to those considered in Section 4. Results for our method when $\delta=1$, are obtained using the ATS expression in \emph{Corollary} \ref{cor:gbe} while the results for $\delta <1$ and MEWMA chart are obtained using 10,000 Monte Carlo simulations.  

The results in Table \ref{tab:compare} shows that under equal expectations, our BTBE chart and the MEWMA chart exhibit similar performance, where the MEWMA with $r=0.1$ is a bit quicker in detecting small shifts, as can be expected. Our method is quickest in detecting large shifts because when a signal is observed on the first event time it does not have to wait until the second event before it can signal. 

Next, consider the unequal expectation case (scenarios 3 and 4), which we believe to be more realistic, as seldom multiple components have equal expected life-times. Here, our method outperformed the MEWMA chart both for $r=1$ and for $r=0.1$. With the exception of the OC-I1 shift in scenario 4, where the $ATS$ values are close but the MEWMA chart is a little bit quicker.

On a side note: we have also tested these two charts for decreasing shifts, where we implement the BTBE with a lower control limit only. The MEMWA chart does not detect decrease shift properly and is outperformed by our proposed chart. We do not include the full comparison as downward shift in exponentially distributed data (like GBE data) is difficult to define (see the discussion in Section \ref{Sec:BTBE_charts})

\emph{Overall}, it is noted that; under equal expected time-between-events and small shift sizes the MEWMA chart with $r=0.1$ has better detection ability. However, under unequal expected time-between-events, our BTBE chart showed significant better performance especially when large shifts are present in the data.

\section{Application to AIDS Data}\label{Sec:case_study}
In this section, we implement the proposed BTBE monitoring method to the AIDS dataset obtained from the Centers for Disease Control (CDC) in Atlanta, Georgia, which is also available in R-package \textbf{SurvTrunc} \citep{SurvTrunc}. In the data, we have a total of 295 people, among which 258 are adults and 37 are children. All people in the sample were infected with AIDS through contaminated blood transfusion. For each person, we have a time to HIV infection (the first event $X_{(1)}$ referred to as \emph{infection time}) and we have the total time to AIDS diagnosis (the second event $X_{(2)}$ referred to as \emph{total incubation time}). We have excluded one person from the data whose event time equals zero and most likely contracted AIDS before the blood transfusion. 

The data is summarized in Table~\ref{tab:desc_stat_data} and visualized in Figure~\ref{AIDS_data_plot}. The infection time ($X_{(1)}$) is significantly higher for children than for adults (at a 5\% significance level). The total incubation time ($X_{(2)}$) is shorter for children than for adults (at a 10\% significance level). These results are inline with \cite{hu2014cross} who concluded that children, compared to adults, have shorter HIV incubation times. Note that we will work with transformed data (division by 100), to have shape and scale parameters of the same size \citep{kundu2009estimating}.

\begin{table}[htbp]
	\centering
	\caption{Descriptive statistics}
	\begin{tabular}{rccc}
		\hline \hline
		&Mean Adults & Mean Children &p-value\\
		\hline
	     Infection time (in months) $X_{(1)}$& 48.7 & 56.9 &0.0127\\
		Total incubation time (in months) $X_{(2)}$&81.2 & 76.2&0.0514\\
		Transformed $X_{(1)}$& 0.487 & 0.569\\
		Transformed $X_{(2)}$&0.812 & 0.762\\	
		\hline\hline
	\end{tabular}
	\label{tab:desc_stat_data}
\end{table}

\begin{figure}[!htb]
	\centering
	\includegraphics[width=0.6\linewidth]{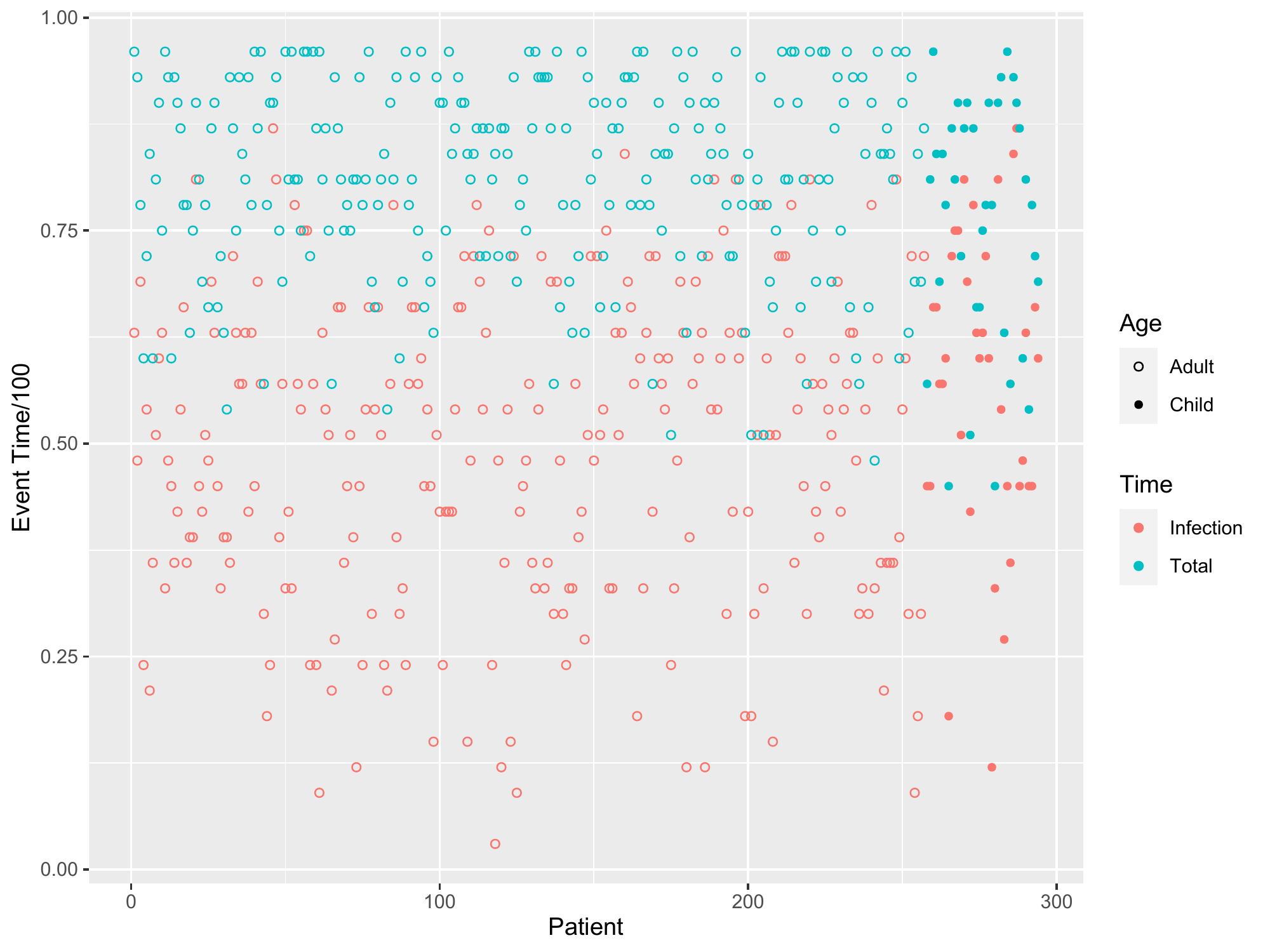}
	\caption{\label{AIDS_data_plot} AIDS data for each subject}
\end{figure}

To illustrate our BTBE chart, we use the adult's data as in-control data and the children's data as shifted data (upward shifted for $X_{(1)}$ and downward shifted for $X_{(2)}$). In this dataset $X_1<X_2$ for all subject, hence it follows that $X_{(1)}=X_1$ and $X_{(2)}=X_2$ for all subjects. 

To implement our chart, we first fit a distribution to the data in a Phase I analysis. We used the R-package \textbf{fitdistrplus} \citep{fitdistrplus} to evaluate various distributions. Figure~\ref{fig:qqplots} shows the Q-Q plots for $X_{(1)}$ and $X_{(2)}$. Note that \cite{marshall1967multivariate} stated that a bivariate dataset fits the MOBW distribution if (i) the marginal distribution of each variable follows a Weibull distribution, and (ii) $X_{(1)}$ follows a Weibull distribution. By the above analysis, it is concluded that a MOBW model best fits our data. 

\begin{figure}[ht]
	\centering
	\includegraphics[width=\linewidth]{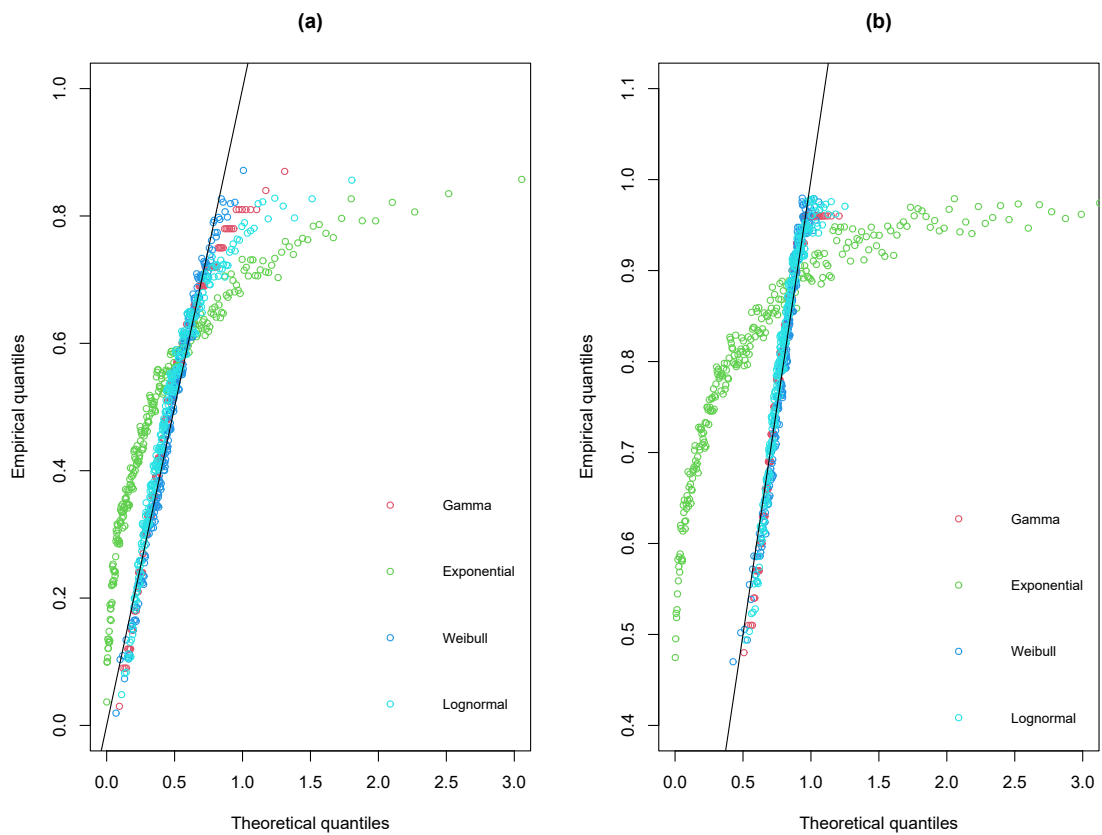}
	\caption{QQ-plot; (a) for $X_{(1)}$, and (b) for $X_{(2)}$}
	\label{fig:qqplots}
\end{figure}

We employ an adjusted EM algorithm to estimate the MOBW parameters, because we only observe $X_1<X_2$. Appendix C provides details on the adjusted EM algorithm. The estimated MOBW parameters of the in-control (adults) data are $\eta=4.31$, $\lambda_1=0.574$ $\lambda_2=0.905$ and  $\lambda_{12}=1.12$. 

Using these estimates, the control limit formulas in Table \ref{tab:cls} yield: $LCL_{(1)} = 0.180$, $UCL_{(1)} = 0.794$ for the first event time and $LCL_{(2)}=\left(x_1^{4.311} + 0.00374\right)^{0.232}$, $UCL_{(2)}=\left(x_1^{4.311} + 2.247\right)^{0.232}$ for the second event time. We have set $ATS_0=25$ because the expected time-between-events is approximately 0.4 and this will yield about one false alarm for approximately each 60 events, i.e. 30 subjects. 

The proposed BTBE chart for the AIDS dataset is plotted in Figure~\ref{fig:BTBE_chart}. The chart plots the univariate superimposed data stream: the first event and the second event of the first subject, followed by the first and second event of the second subject etc. The chart shows that the first event signals five times (the small "1" indicates that the signal is related to a first event). This is in line with the results of Table \ref{tab:desc_stat_data}. 

\begin{figure}[!htb]
	\centering
	\includegraphics[width=\linewidth]{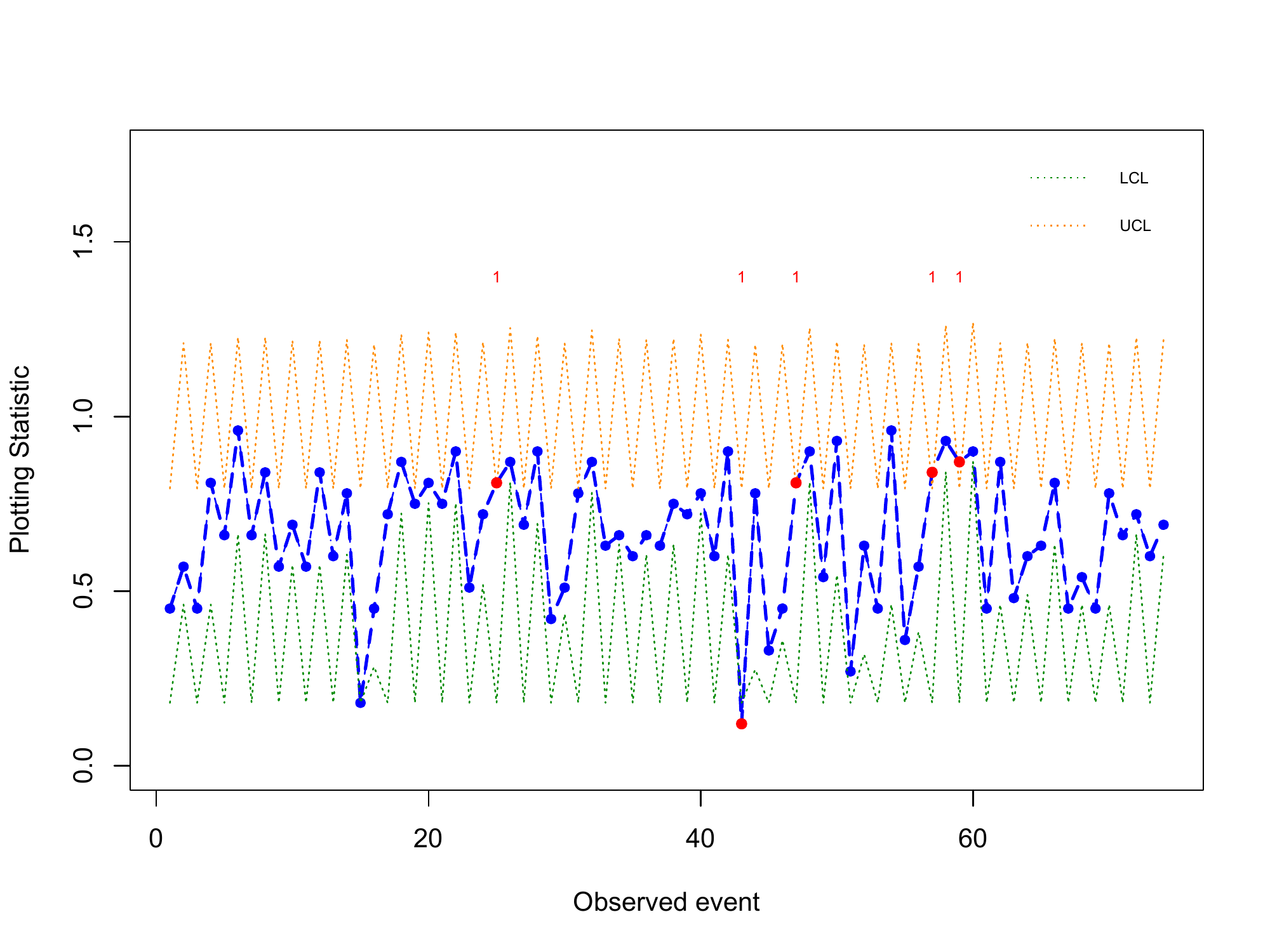}
	\caption{The BTBE chart for the AIDS case study}
	\label{fig:BTBE_chart}
\end{figure}

\section{Concluding Remarks and Recommendations}\label{Sec:conclusions}
This article proposed a novel multivariate time-between-event method referred to as the BTBE chart. The BTBE chart has real-time detection power and does not have a built-in delay like all existing multivariate time-between-event methods. Numerical results and theoretical arguments showed that it is a realistic way to monitor multivariate time-between-event data, and has better detection ability as compared to existing MEWMA method. Our method performs excellently especially when the observations have unequal time-between-events, which we consider a more realistic scenario than equal expectations. Under equal expectations, we have comparable performance. 

Future interesting work would be to extend the method to point-process data, our method is only applicable for the monitoring of vector-based event data. Also extending our bivariate chart to multivariate data will be of interest. In the current methodology, the history of the events is not accumulated so, a method based on EWMA and CUSUM type structures will also be an interesting issue for future research. A potential to fine tune our method is to have a signal once X gets bigger than the UCL rather than to wait until we observe X.

\section*{Supplementary Material}
The supplimentry material of this manuscript, which consists of some proofs is also provided in form of pdf.

\section*{Acknowledgement}
The work of Inez M. Zwetsloot described in this paper was partially supported by a grant from the Research Grants Council of the Hong Kong Special Administrative Region, China (Project No. CityU 21215319) and by a grant from City University of Hong Kong (Project No. 7005567). 

\bibliography{references}

\newpage
\section*{Appendix A} \label{app:A}
This Appendix provides the proofs for Theorems \ref{thm:F_x1}, \ref{thm:F_x2} and \ref{thm:ats}. 

\subsection*{A.1 Proof Theorem 1}
This proof is based on censored data modeling. We follow the same steps as \cite{wienke2010frailty}.
\begin{proof}[Proof of Theorem \ref{thm:F_x1}]
	We derive the conditional cumulative distribution function for the first observed event time $X_{(1)}$ for each case $X_1 < X_2$, $X_1 > X_2$, and $X_1 = X_2$, separately. 
	
	When $X_1 < X_2$, the conditional CDF is defined as $F_{X_{(1)}}(x_{(1)}) =P(X_1 \leq x_{(1)}, X_1<X_2) /P(X_1 < X_2)$, for which
	\begin{align*}
	P(X_1 < X_2 ) &= \int_{0}^{\infty}\int_{x_1}^{\infty} f ( x_1,x_2 )\;dx_2dx_1 
	= \int_{0}^{\infty} \int_{x_1}^{\infty} \frac{\partial}{\partial x_2} S_1 (x_1,x_2) \; dx_2dx_1 \\
	&= \int_{0}^{\infty} S_1 (x_1,x_2) \bigg\rvert^{x_2=\infty}_{x_2=x_1} \; dx_1 
	=  - \int_{0}^{\infty} S_1 (x_1,x_1) \; dx_1 
	\end{align*}
	where $f(x_1,x_2)= \frac{\partial^2 }{\partial x_1 \partial x_2} S(x_1,x_2)=\frac{\partial }{ \partial x_2} S_1(x_1,x_2)$ and $ S_1 (x_1,\infty) = 0$. We also obtain 
	\begin{align*}
	P(X_1 \leq x_{(1)}, X_1 < X_2) & =\int_{0}^{x_{(1)}}\int_{x_1}^{\infty} f( x_1,x_2 )\;dx_2dx_1 
	= \int_{0}^{x_{(1)}} \int_{x_1}^{\infty} \frac{\partial}{\partial x_2} S_1 (x_1,x_2) \; dx_2dx_1 \\
	&= \int_{0}^{x_{(1)}} S_1 (x_1,x_2) \bigg\rvert^{x_2=\infty}_{x_2=x_1} \; dx_1 
	=  - \int_{0}^{x_{(1)}} S_1 (x_1,x_1) \; dx_1 . 
	\end{align*}
	Therefore, $F_{X_{(1)}}(x_{(1)})= \frac{\bigintsss_{0}^{x_{(1)}} S_1 (x_1,x_1) \; dx_1}{ \bigintsss_{0}^{\infty} S_1 (x_1,x_1) \; dx_1 }$ if $X_1 < X_2$. A symmetric argument gives $F_{X_{(1)}}(x_{(1)})= \frac{\bigintsss_{0}^{x_{(1)}} S_2 (x_2,x_2) \; dx_2}{ \bigintsss_{0}^{\infty} S_2 (x_2,x_2) \; dx_2 }$ if $X_1>X_2$.
	
	When $X_1=X_2$, the conditional CDF is defined as $F_{X_{(1)}}(x_{(1)}) = P(X_{(1)} \leq x_{(1)}, X_1=X_2)/P(X_2 = X_1)$, where $P(X_1 \leq x_{(1)}, X_1=X_2) = \bigintsss_{0}^{x_{(1)}} f(x_1,x_1) \; dx_1 $ and $P(X_1 = X_2) = \bigintsss_{0}^{\infty} f (x_1,x_1) \; dx_1 $. \\
	
	Similarly, the conditional survival function is defined as $S_{X_{(1)}}(x_{(1)}) =P(X_1 > x_{(1)}, X_1<X_2) /P(X_1 < X_2)$ if $X_1 < X_2$, for which
	\begin{align*}
		P(X_1 > x_{(1)}, X_1 < X_2) &= \int_{x_{(1)}}^{\infty}\int_{x_1}^{\infty} f( x_1,x_2 )\;dx_2dx_1 
		= \int_{x_{(1)}}^{\infty}\int_{x_1}^{\infty} \frac{\partial}{\partial x_2} S_1 (x_1,x_2) \; dx_2dx_1 \\
		&= \int_{x_{(1)}}^{\infty} S_1 (x_1,x_2) \bigg\rvert^{x_2=\infty}_{x_2=x_1} \; dx_1 
		=  - \int_{x_{(1)}}^{\infty} S_1 (x_1,x_1) \; dx_1 .
	\end{align*}
	So it follows that $S_{X_{(1)}}(x_{(1)}) = \frac{\bigintsss_{x_{(1)}}^{\infty} S_1 (x_1, x_1) dx_1}{ \bigintsss_0^{\infty} S_1 (x_1, x_1) dx_1}$ if $X_1<X_2$. Equivalently, we obtain $ S_{X_{(1)}}(x_{(1)}) = P(X_2 >  x_{(1)}, X_2<X_1)/P(X_2 < X_1 ) = \frac{\bigintsss_{x_{(1)}}^{\infty} S_2(x_2, x_2) dx_2 }{ \bigintsss_0^{\infty} S_2 (x_2, x_2) dx_2}$ if $X_1 > X_2 $ and $S_{X_{(1)}}(x_{(1)}) =  P(X_1 > x_{(1)}, X_1=X_2)/P(X_1 = X_2) = \frac{\bigintsss_{x_{(1)}}^{\infty} f(x) dx_1}{ \bigintsss_0^{\infty} f(x) dx_1}$ if $X_1 = X_2$.
	This completes the proof.
\end{proof}

\subsection*{A.2 Proof Theorem 2}
\noindent
For the proof of Theorem~\ref{thm:F_x2}, we need the following three lemmas.
\begin{lemma}\label{lem:f_x}
	The probability density function of $X_{(1)}$ is given as
	\begin{equation}
    f_{X_{(1)}}(x_{(1)}) = 
    \begin{cases}
				-S_1 (x_1, x_1) \;\;\; \text{if}\; X_1 < X_2\\
				-S_2 (x_2, x_2) \;\;\; \text{if}\; X_1 > X_2.
	\end{cases}
	\end{equation}
\end{lemma}

\begin{proof}[Proof of Lemma~\ref{lem:f_x}]
	From Theorem \ref{thm:F_x1}, we know that $P(X_{(1)} \leq x_{(1)}, X_1 < X_2) = - \int_{0}^{x_{(1)}} S_1 (x_1,x_1) \;dx_1$. The probability density function of $X_{(1)}$ when $X_1 < X_2$ is
	\begin{align*}
	f_{X_{(1)}}(x_1)=\dfrac{d}{dx_1}\; \int_{0}^{x_1} -S_1 (x,x) \;dx = -S_1 (x_1,x_1).
	\end{align*} 
	Similarly, for $X_1 > X_2$, we obtained $f_{X_{(1)}}(x_2)=- S_2 (x_2,x_2)$. This completes the proof.
\end{proof}

\begin{lemma} \label{lem:F_X_1}
	The partial distribution function of $(X_{(1)}, X_{(2)})$ with respect to $X_{(1)}$ is defined as $ F_{X_{(1)}}(x_{(1)}, x_{(2)}) = P(X_{(1)}=x_{(1)}, X_{(2)}\leq x_{(2)})$ which is equal to
	\begin{equation}
    F_{X_{(1)}}(x_{(1)}, x_{(2)}) = 
    \begin{cases}
				S_1 (x_1, x_2)-S_1 (x_1, x_1) \;\;\; \text{if}\; X_1 < X_2\\
				S_2 (x_1, x_2)-S_2 (x_2, x_2) \;\;\; \text{if}\; X_1 > X_2.
	\end{cases}
	\end{equation}
\end{lemma}

\begin{proof}[Proof of Lemma~\ref{lem:F_X_1}]
	For $X_1 < X_2$, it follows that 
	\begin{align*}
	F_{X_{(1)}}(x_{(1)}, x_{(2)}) &= P(X_1=x_1, X_2 \leq x_2, X_1 < X_2)=  P(X_1=x_1, x_1< X_2 \leq x_2)\\
	&= \int_{x_1}^{x_2} f(x_1,x_2)dx_2
	= \int_{x_1}^{x_2} \frac{\partial}{\partial x_2} S_{1}(x_1,x_2)dx_2 
	= S_1 (x_1,x_2) \bigg\rvert^{x_2=x_2}_{x_2=x1} \\
	&=S_{1}(x_1,x_2) - S_{1}(x_1,x_1).
	\end{align*}
	
	Equivalently, when $X_1 > X_2$, it follows that  $F_{X_{(1)}}(x_{(1)},x_{(2)}) =S_{2}(x_1,x_2) - S_{2}(x_2,x_2)$. This completes the proof.
\end{proof}

\begin{lemma} \label{lem:S_X_1}
	The partial survival function of $(X_{(1)}, X_{(2)})$ with respect to $X_{(1)}$ is defined as $S_{X_{(1)}}(x_{(1)}, x_{(2)}) = 	P(X_{(1)}=x_{(1)}, X_{(2)} > x_{(2)})$ and is equal to 
	\begin{equation}
    S_{X_{(1)}}(x_{(1)}, x_{(2)}) = 
    \begin{cases}
				 -S_1 (x_1, x_2) \;\;\; \text{if}\; X_1 < X_2\\
				 -S_2 (x_1, x_2) \;\;\; \text{if}\; X_1 > X_2.
	\end{cases}
	\end{equation}
\end{lemma}

\begin{proof}[Proof of Lemma~\ref{lem:S_X_1}]
	For $X_1 < X_2$ it follows that
	\begin{align*}
	S_{X_{(1)}}(x_{(1)},x_{(2)}) &= P(X_1=x_1, X_2 > x_2, X_1 < X_2)\\
		 &= \int_{x_2}^\infty f(x_1,x_2)dx_2
	= \int_{x_2}^\infty \frac{\partial}{\partial x_2} S_{1}(x_1,x_2)dx_2\\
	&= S_1 (x_1,x_2) \bigg\rvert^{x_2=\infty}_{x_2=x_2}
	= - S_{1}(x_1,x_2).
	\end{align*}
	Equivalently when $X_1 > X_2$, it follows that  $S_{X_{(1)}}(x_{(1)},x_{(2)}) = - S_{2}(x_1,x_2)$. This completes the proof.
\end{proof}


\begin{proof}[Proof of Theorem~\ref{thm:F_x2}]
	By definition, the conditional distribution function is
	$$
	F_{X_{(2)}|X_{(1)}}(x_{(2)}|x_{(1)}) = P[ X_{(2)} \leq x_{(2)}|X_{(1)}=x_{(1)}]= \frac{P[X_{(1)}=x_{(1)}, X_{(2)} \leq x_{(2)}]}{P[X_{(1)}=x_{(1)}]}
	=  \frac{F_{X_{(1)}}(x_{(1)},x_{(2)})}{f_{X_{(1)}}(x_{(1)})}.
	$$
	From lemmas \ref{lem:f_x} and \ref{lem:F_X_1} it follows that 
	\begin{equation}
    F_{X_{(2)}|X_{(1)}}(x_{(2)}| x_{(1)}) = 
    \begin{cases}
				1-\frac{S_1(x_1, x_2)}{S_1(x_1, x_1)} \;\;\; \text{if}\; X_1 < X_2\\
				1-\frac{S_2(x_1, x_2)}{S_2(x_2, x_2)} \;\;\; \text{if}\; X_1 > X_2.
	\end{cases}
	\end{equation}
\noindent
Similarly, the conditional survival function is
	$$
	S_{X_{(2)}|X_{(1)}}(x_{(2)}|x_{(1)}) = \frac{P(X_{(1)}=x_{(1)}, X_{(2)} > x_{(2)})}{P(X_{(1)}=x_{(1)})}
	=  \frac{S_{X_{(1)}}(x_{(1)},x_{(2)})}{f_{X_{(1)}}(x_{(1)})}.
	$$
	From lemmas \ref{lem:f_x} and \ref{lem:S_X_1} it follows that 
	\begin{equation}
    S_{X_{(2)}|X_{(1)}}(x_{(2)}|x_{(1)}) = 
    \begin{cases}
				\frac{S_1(x_1, x_2)}{S_1(x_1, x_1)} \;\;\; \text{if}\; X_1 < X_2\\
				\frac{S_2(x_1, x_2)}{S_2(x_2, x_2)} \;\;\; \text{if}\; X_1 > X_2.
	\end{cases}
	\end{equation}
	This completes the proof.
\end{proof}

\subsection*{A.3 Proof Theorem 3}
This section provides the proof of Theorem~\ref{thm:ats}. For the proof we need the following geometric series:

\begin{definition}\label{def:gs}
The following holds when $|r|<1$
\begin{align} \label{eq:gs}
	\sum_{t=k}^{\infty} \frac{t!}{(t-k)!} r^{t-k} = \frac{k!}{(1-r)^{k+1}} \;\; \text{for} \; k=0,1,2,3,... .
\end{align}
For $k=0$ this simplifies to the standard geometric series $\sum_{t=0}^{\infty} r^{t} = \frac{1}{1-r}$ and for $k=1$ we have $\sum_{t=1}^{\infty} t r^{t-1} = \frac{1}{(1-r)^2}$.
Some mathematical manipulation of Equation \eqref{eq:gs} results in the following series, which holds when $|r|<1$:
\begin{align*}
	&\sum_{t=k}^{\infty} t \frac{\left( t-\frac{k+1}{2} \right)!}{\left(t-k\right)!} r^{t-k} = \frac{\left(\frac{k-1}{2}\right)! \times \frac{k-1}{2}}{(1-r)^{\frac{k+1}{2}}} + \frac{\left(\frac{k+1}{2}\right)!}{(1-r)^{\frac{k+1}{2}+1}} \;\; \text{for } k=1,3,5,7,...\\
	&\sum_{t=k}^{\infty} t \frac{\left( t-\frac{k-2}{2} \right)!}{\left(t-k\right)!} r^{t-k} = \frac{\left(\frac{k}{2}\right)! }{(1-r)^{\frac{k}{2}}} + \frac{\left(\frac{k}{2}\right)!}{(1-r)^{\frac{k}{2}+1}} \;\;\;\;\;\;\;\; \text{for } k=2,4,6,...
\end{align*}
\end{definition}

\noindent
From definition \ref{def:gs} it follows that for $k=1,3,5$ and $|r|<1$,
	\begin{align}\label{eq:gs1}
		\begin{split}
		&\sum_{t=1}^{\infty} t r^{t-1} = \frac{0}{1-r} + \frac{1}{(1-r)^{2}}\\
		&\sum_{t=3}^{\infty} t(t-2) r^{t-3} = \frac{1}{(1-r)^2} + \frac{2!}{(1-r)^{3}}\\
		&\sum_{t=5}^{\infty} t(t-3)(t-4) r^{t-5} = \frac{2! \times 2}{(1-r)^{3}} + \frac{3!}{(1-r)^{4}}
	\end{split}
	\end{align}
	and that for $k=2,4,6$, 
	\begin{align} \label{eq:gs2}
		\begin{split}
		&\sum_{t=2}^{\infty} t r^{t-2} = \frac{1}{1-r} + \frac{1}{(1-r)^2}\\
		&\sum_{t=4}^{\infty} t(t-3) r^{t-4} = \frac{2}{(1-r)^{2}} + \frac{2}{(1-r)^{3}}\\
		&\sum_{t=6}^{\infty} t(t-4)(t-5) r^{t-6} = \frac{3!}{(1-r)^{3}} + \frac{3!}{(1-r)^{4}}
		\end{split}.
	\end{align}

\begin{proof}[Proof of Theorem~\ref{thm:ats}]
We use the shorthand notation as defined in Equations \eqref{eq:prob1} and \eqref{eq:prob2}, where $=$ and $\neq$ represent $X_{1}=X_{2}$ and $X_{1}\neq X_{2}$, respectively. Also, $S_1$ and $NS_1$ indicate a signal or no signal on the first event while $S_2$ and $NS_2$ indicate a signal or no signal on the second event. 

Recall that $ATS = ARL *E[TBE]$, in Equation (\ref{eq:ETBE}) we gave the expressions for $E[TBE]$ so we only need to derive an expression for the $ARL$. By definition $ARL = \sum_{i=1}^{\infty} i *P[\text{Signal at event }i]$. In our scenario we have two observations that can signal: we can get a signal on the first event or on the second event. Hence we can split our signal probability into:
\begin{align*}
P[\text{Signal at event }i] &= P[\text{No signal for all events upto } i-1]P[S_1] \\
&+ P[\text{No Signal for all events upto } i-2]P[NS_1,S_2,\neq].
\end{align*}
Hence, we get 
\begin{align}
	ARL &= \sum_{i=1}^{\infty} i *P[\text{No signal for all events upto } i-1]P[S_1]  \label{eq:s1}\\
	&+ \sum_{i=2}^{\infty} i *P[\text{No Signal for all events upto } i-2]P[NS_1,S_2,\neq] \label{eq:ns1s2}
\end{align}

There are two signal scenarios ($S1$ or $S2$), first we focus on the probability of observing a signal at the first event time $P[S_{1}]$ (Equation~\eqref{eq:s1}). We need the probability that we did not observed a signal upto event $i-1$. There are two no-signal scenarios: either we do not observe a signal with probability $P[NS_1, =]$ when $X_1=X_2$, or with probability $P[NS_{1}, NS_{2}, \neq]$ when $X_{1} \neq X_{2}$. We include all possible combinations of these two to obtain $i-1$ events without a signal. First, we consider only observing $[NS_1,=]$ for all $i-1$ events then, we consider the probability that we have one event $[NS_1,NS_2,\neq]$ and the other $i-3$ events are $[NS_1,=]$. We do this for all possible combinations and obtain the following summation:

\begin{align*}
	\sum_{i=1}^{\infty} i *P &[\text{No signal for all events upto } i-1]P[S_1] = \\
	& \sum_{i=1}^{\infty}iP[S_{1}]P[NS_{1}, =]^{i-1} \\
	&+\sum_{i=3}^{\infty}i (i-2) P[S_{1}] P[NS_{1}, NS_{2}, \neq] P[NS_{1}, =]^{i-3} \\
	&+\sum_{i=5}^{\infty}i \frac{(i-3)!}{2!(i-5)!} P[S_{1}] P[NS_{1}, NS_{2}, \neq]^{2} P[NS_{1}, =]^{i-5}\\
	&+\sum_{t=7}^{\infty}t \frac{(t-4)!}{3!(t-7)!} P[S_{1}] P[NS_{1}, NS_{2}, \neq]^{3} P[NS_{1}, =]^{t-7}\\
	&+......,
\end{align*}

Next, we focus on a signal at the second event time $[NS_{1}, S_{2}, \neq]$ in Equation~\eqref{eq:ns1s2}. Similarly, there are two no-signal scenarios, $[NS_1, =]$ or $[NS_1, NS_2, \neq]$. We include all combinations of signal and no-signal scenarios and we get 
\begin{align*}
	\sum_{i=2}^{\infty} i *P&[\text{No signal for all events upto } i-2]P[NS_1,S_2,X_1\neq X_2] = \\
	&\sum_{i=2}^{\infty}i P[NS_{1}, S_{2}, \neq] P[NS_{1}, =]^{i-2}\\
	&+\sum_{i=4}^{\infty}i (i-3) P[NS_{1}, S_{2}, \neq] P[NS_{1}, NS_{2}, \neq] P[NS_{1},=]^{i-4}\\
	&+\sum_{i=6}^{\infty}i \frac{(i-4)!}{2!(i-6)!} P[NS_{1}, S_{2}, \neq] P[NS_{1}, NS_{2}, \neq]^{2} P[NS_{1}, =]^{i-6}\\
	&+......,
\end{align*}

Adding up these two components, using Definition \ref{def:gs} and Equations \eqref{eq:gs1} and \eqref{eq:gs2} the $ARL$ is equal to

\begin{align*}
	ARL = &P[S_{1}] \left(\frac{0}{1-P[NS_{1}, =]} + \frac{1}{(1-P[NS_{1}, =])^{2}}\right)\\
	& + P[S_{1}] P[NS_{1}, NS_{2}, \neq] \left(\frac{1}{(1-P[NS_{1}, =])^{2}} + \frac{2}{(1-P[NS_{1}, =])^{3}}\right)\\
	& + P[S_{1}] P[NS_{1}, NS_{2}, \neq]^{2} \left(\frac{2}{(1-P[NS_{1}, =])^{3}} + \frac{3}{(1-P[NS_{1}, =])^{4}}\right)\\
	& + ......\\
	& + P[NS_{1}, S_{2}, \neq] \left(\frac{1}{1-P[NS_{1}, =]} + \frac{1}{(1-P[NS_{1}, =])^{2}}\right)\\
	& + P[NS_{1}, S_{2}, \neq] P[NS_{1}, NS_{2}, \neq] \left(\frac{2}{(1-P[NS_{1}, =])^{2}} + \frac{2}{(1-P[NS_{1}, =])^{3}}\right)\\
	& + P[NS_{1}, S_{2}, \neq] P[NS_{1}, NS_{2}, \neq]^{2} \left(\frac{3}{(1-P[NS_{1}, =])^{3}} + \frac{3}{(1-P[NS_{1}, =])^{4}}\right)\\
	& + ......
\end{align*}
This can also be written as
\begin{align*}
	ARL = &\sum_{i=0}^{\infty}P[S_{1}] P[NS_{1}, NS_{2}, \neq]^{i}\left(\frac{i}{(1-P[NS_{1}, =])^{i+1}} + \frac{i+1}{(1-P[NS_{1}, =])^{i+2}}\right)\\
	& + \sum_{i=0}^{\infty}P[NS_{1}, S_{2}, \neq] P[NS_{1}, NS_{2}, \neq]^{i} \left(\frac{i+1}{(1-P[NS_{1}, =])^{i+1}} + \frac{i+1}{(1-P[NS_{1}, =])^{i+2}}\right).
\end{align*}
By carrying out some mathematical manipulations and applying Equation \eqref{eq:gs} for $k=0,1$ (the traditional geometric series), we get the following expression for the $ARL$: 
\begin{align*}
	ARL  = &\frac{P[S_1]}{1-P[NS_{1}, =]}\sum_{i=0}^{\infty} i \left(\frac{P[NS_{1}, NS_{2}, \neq]}{1-P[NS_{1}, =]}\right)^{i} \\
	&+ \frac{P[S_1]}{(1-P[NS_{1}, =])^{2}}\sum_{i=0}^{\infty} (i+1)\left(\frac{P[NS_{1}, NS_{2}, \neq]}{1-P[NS_{1}, =]}\right)^{i} \\
	&+ \frac{P[NS_1,S_2,\neq]}{1-P[NS_{1}, =]}\sum_{i=0}^{\infty} (i+1) \left(\frac{P[NS_{1}, NS_{2}, \neq]}{1-P[NS_{1}, =]}\right)^{i} \\
	&+ \frac{P[NS_1,S_2,\neq]}{(1-P[NS_{1}, =])^{2}}\sum_{i=0}^{\infty} (i+1)\left(\frac{P[NS_{1}, NS_{2}, \neq]}{1-P[NS_{1}, =]}\right)^{i} \\
	= &\frac{P[S_1]}{1-P[NS_{1}, =]}\sum_{i=1}^{\infty} (i-1) \left(\frac{P[NS_{1}, NS_{2}, \neq]}{1-P[NS_{1}, =]}\right)^{i-1}\\ 
	&+ \frac{P[S_1]}{(1-P[NS_{1}, =])^{2}}\sum_{i=1}^{\infty} i \left(\frac{P[NS_{1}, NS_{2}, \neq]}{1-P[NS_{1}, =]}\right)^{i-1} \\
	&+ \left(\frac{P[NS_1,S_2,\neq]}{1-P[NS_{1}, =]}+ \frac{P[NS_1,S_2,\neq]}{(1-P[NS_{1}, =])^{2}}\right)\sum_{i=1}^{\infty} i \left(\frac{P[NS_{1}, NS_{2}, \neq]}{1-P[NS_{1}, =]}\right)^{i-1}\\
	= &\frac{P[S_{1}]P[NS_{1}, NS_{2}, \neq]+P[S_{1}]+(2-P[NS_{1},=])P[NS_{1}, S_{2}, \neq]}{(P[S_{1}]+P[NS_{1}, S_{2}, \neq])^{2}}\\
	= &\frac{1+P[NS_{1}, \neq]} {P[S_{1}]+P[NS_{1}, S_{2}, \neq]}.
\end{align*}

\noindent
Therefore,
$$
ATS^{OC} = \frac{1 + P_*[X_1 \neq X_2]P_*[NS_1|X_1\neq X_2]}{P_*[S_1] + P_*[X_1\neq X_2]P_*[NS_1,S_2| X_1\neq X_2]} E_*[TBE]
$$
This completes the proof. 
\end{proof}

\section*{Appendix B}
In this Appendix we provide more details on the selected bivariate lifetime distributions: the GBE, MOBE, and MOBW models. Table \ref{tab:par} gives the values of the parameters to obtain the synthetic data as used for the performance analysis in Table \ref{tab:ats}. For each model we discuss the survival function, estimation of the model parameters and simulation of random variables. In addition, Tables \ref{tab:gbe}-\ref{tab:mobw} provide the distribution functions and the probabilities of having $X_1<X_2$, $X_1>X_2$ or $X_1=X_2$, as well as expressions for the expected values of $X_1$, $X_2$ and the time-between-events. 

\begin{table}[ht]
	\centering
	\footnotesize
	\begin{tabular}{llllllllllllll}
		\hline \hline
		&&&&\multicolumn{3}{c}{GBE } &\multicolumn{3}{c}{MOBE} &\multicolumn{4}{c}{MOBW }  \\
		\cline{5-7} \cline{8-10} \cline{11-14}
		Scen&Shift&$E[X_1]$ &$E[X_2]$ &$\theta_1$&$\theta_2$&$\delta$ &$\lambda_1$&$\lambda_2$&$\lambda_{12}$&$\lambda_1$&$\lambda_2$&$\lambda_{12}$&$\eta$ \\
		\hline
		1. &IC&5&5&5&5&1&0.2&0.2&0&0.0314&0.0314&0&2\\
		&OC-I1&7.5&5&7.5&5&1&0.133&0.2&0&0.0140&0.0314&0&2\\
		&OC-I1&10&5&10&5&1&0.1&0.2&0&0.0079&0.0314&0&2\\
		&OC-I2&7.5&7.5&7.5&7.5&1&0.133&0.133&0&0.0140&0.0140&0&2 \\
		&OC-I2&10&10&10&10&1&0.1&0.1&0&0.0079&0.0079&0&2\\
		&OC-D1&2.5&5&*&*&*&*&*&*&0.1257&0.0314&0&2\\
		&OC-D2&2.5&2.5&*&*&*&*&*&*&0.1257&0.1257&0&2\\
		\hline
		2.&IC&5&5&5&5&0.5&0.164&0.164&0.036&0.0257&0.0257&0.0057&2\\
		&OC-I1&7.5&5&7.5&5&0.5&0.103&0.170&0.030&0.0098&0.0273&0.0041&2\\
		&OC-I1&10&5&10&5&0.5&0.073&0.173&0.027&0.0043&0.0278&0.0036&2\\
		&OC-I2&7.5&7.5&7.5&7.5&0.5&0.109&0.109&0.024&0.0114&0.0114&0.0025&2\\
		&OC-I2&10&10&10&10&0.5&0.081&0.081&0.018&0.0064&0.0064&0.0014&2\\
		&OC-D1&2.5&5&*&*&*&*&*&*&0.1114&0.0171&0.0143&2\\
		&OC-D2&2.5&2.5&*&*&*&*&*&*&0.1028&0.1028&0.0228&2\\
		\hline
		3.  &IC &5&15&5&15&1&0.2&0.067&0&0.0314&0.0035&0&2\\
		&OC-I1&7.5&15&7.5&15&1&0.133&0.067&0&0.0140&0.0035&0&2\\
		&OC-I1&10&15&10&15&1&0.1&0.067&0&0.0079&0.0035&0&2\\
		&OC-I2&7.5&22.5&7.5&22.5&1&0.133&0.044 &0&0.0140&0.0016&0&2 \\
		&OC-I2&10&30&10&30&1&0.1&0.033&0&0.0079&0.0009&0&2 \\
		&OC-D1&2.5&15&*&*&*&*&*&*&0.1257&0.0035&0&2\\
		&OC-D2&2.5&7.5&*&*&*&*&*&*&0.1257&0.0140&0&2\\
		\hline
		4. &IC &5&15&5&15&0.5&0.176&0.042&0.024&0.0282&3.17e-04&0.0032&2 \\
		&OC-I1&7.5&15&7.5&15&0.5&0.115&0.048&0.018&0.0124&1.90e-03&0.0016&2 \\
		&OC-I1&10&15&10&15&0.5&0.085&0.052&0.015&0.0068&2.46e-03&0.0010&2 \\
		&OC-I2&7.5&22.5&7.5&22.5&0.5&0.117&0.028&0.016&0.0126&1.41e-04&0.0014&2  \\
		&OC-I2&10&30&10&30&0.5&0.088&0.021&0.012&0.0070&7.93e-05&0.0008&2 \\
		&OC-D1&2.5&15&*&*&*&*&*&*&0.1139&8.25e-03&0.0117&2 \\
		&OC-D2&2.5&7.5&*&*&*&*&*&*&0.1130&1.26e-03&0.0127&2\\		 
		\hline\hline
	\end{tabular}
	\label{tab:par}
	\caption{Parameter values for GBE, MOBE and MOBW models for synthetic data generation}
\end{table}

\subsection*{B.1 Gumbel's Bivariate Exponential distribution}

The Gumbel’s Bivariate Exponential (GBE) distribution is the most well known model, which was first introduced by \cite{gumbel1960}. The GBE model assumed a failure mechanism driven by a random external stress factor. \cite{gumbel1960} provided two types of GBE models while \cite{hougaard1986} extended GBE type B model. Its survival function with parameters $ \theta_1,\theta_2, \delta$ is 
\begin{equation*} \label{eq:sf.GBE}
	S(x_1,x_2) = \exp\left( \left( \left(\frac{x_1}{\theta _1}\right)^{1/\delta}+ \left(\frac{x_2}{\theta_2}\right)^{1/\delta} \right)^\delta \right),
\end{equation*}

For parameter estimation of the GBE model parameters, one can derive the maximum likelihood estimators as:
$\hat{\theta}_1 = \Bar{x}_{1} = n^{-1} \sum_{i=1}^{n}x_{1i}$, $\hat{\theta}_2 = \Bar{x}_{2} = n^{-1} \sum_{i=1}^{n}x_{2i}$ and $\hat{\delta} = -(log2)^{-1} n^{-1} \sum_{i=1}^{n}\min\{x_{1i}/\Bar{x}_{1}, x_{2i}/\Bar{x}_{2}\}$. For more details on deriving these estimates, the reader is referred to \cite{lu1991inference}. 

In order to simulate data from the GBE model, first obtain $Q$ a uniform random variable (i.e., $Q\sim U(0,1)$) and $R=R_1+N_{\delta}R_2$ where $R_1$ and $R_2$ follows an exponential distribution with unit mean. Furthermore, $N_{\delta}=0$ with probability $1-\delta$ and $N_{\delta}=1$ with probability $\delta$. Next, compute $X_1=\theta_1 {Q^\delta} R$ and $X_2=\theta_2 {Q^{1-\delta}} R$. Table \ref{tab:gbe} provides more details on GBE distributed event data.

\begin{table}[h]
\caption{Characteristics of the GBE distribution}
\centering
\begin{tabular}{lrl}
	
\hline \hline

pdf &$f(x_1,x_2) =$& {\tiny
 $\left(\frac{x_1}{\theta_1}\right)^{(1/\delta) - 1 }\left(\frac{x_2}{\theta_2} \right)^{(1/\delta) - 1} C(x_1,x_2)^{\delta-2}  (C(x_1,x_2)^\delta + \frac{1}{\delta}-1) \exp(-C(x_1,x_2)^\delta)$} \\
&$C(x_1,x_2) =$&$ \left( \frac{x_1}{\theta_1}\right)^{1/\delta} +\left(\frac{x_2}{\theta_2} \right)^{1/\delta} $\\
\hline
 Survival function &$S(x_1,x_2)=$& $\exp\left(-C(x_1,x_2)^{\delta}\right)$ \\
\hline
Expectations &$E[X_1] =$&$\theta_1$ \\
&$E[X_2] =$&$\theta_2$ \\
&$E[X_{(1)}] =$&$C(1,1)^{-\delta}$\\
&$E[X_{(2)}] =$&$ \theta_1+\theta_2-C(1,1)^{-\delta}$ \\
&$E[TBE]=$&$0.5(\theta_1+\theta_2-C(1,1)^{-\delta})$ \\
\hline
Probabilities&$P[X< Y] =$ &$\frac{\theta_1^{-1/\delta}}{C(1,1)}$\\
&$P[X>Y] =$ &$\frac{\theta_2^{-1/\delta}}{C(1,1)}$\\
\hline \hline
\end{tabular}
\label{tab:gbe}
\end{table}

\subsection*{B.2 Marshall Olkin Bivariate Exponential distribution}
The MOBE model was built to model the life-time of a system with two-components which is affected by external shocks. The survival function with parameters $\lambda_1,\lambda_2, \lambda_{12} >0 $ is 
\begin{equation*}
	S(x_1,x_2) = \exp(-\lambda_1 x_1 - \lambda_2 x_2 - \lambda_{12} \max(x_1,x_2)), \,\,\, x_1, x_2 > 0
\end{equation*}
For parameter estimation of the MOBE model, one can derive the maximum likelihood estimations by solving the following maximum likelihood equations:
$n_1/\hat{\lambda}_1 + n_2/(\hat{\lambda}_1+\hat{\lambda}_{12}) = \sum x_{1, i}$, $n_1/(\hat{\lambda}_2 + \hat{\lambda}_{12})+n_2/\hat{\lambda}_2 = \sum x_{2, i}$ and $ n_1/(\hat{\lambda}_2 + \hat{\lambda}_{12}) + n_2/(\hat{\lambda}_1+\hat{\lambda}_{12}) + n_3/\hat{\lambda}_{12} = \sum \max (x_{1, i}, x_{2, i})$
where $n_1, n_2, n_3$ are the number of observations in the regions $X_1 <X_2$, $X_1 > X_2$, and $X_1=X_2$, respectively. For more details on estimation of the MOBE model, the reader is referred to \cite{bemis1972estimation,bhattacharyya1973test,proschan1976estimating}.

In order to simulate data according to the MOBE model, one first obtains $P, Q$ and $R$ as independent exponential distributed variables with mean $\lambda^{-1}_1$, $\lambda^{-1}_2$ and $\lambda^{-1}_{12}$, respectively. Next, compute $X_1=min(P,R)$ and $X_2=min(Q,R)$. Table \ref{tab:mobe} provides more details on MOBE distributed event data.
\begin{table}[h]
\caption{Characteristics of the MOBE distribution}
\centering
\begin{tabular}{lrl}
\hline \hline
pdf &$f(x_1,x_2) =$& 
$\begin{cases}
\lambda_{1}(\lambda_{2}+\lambda_{12}) \exp(-\lambda_{1}x_1-(\lambda_{2}+\lambda_{12})x_2) & x_1<x_2\\
\lambda_{2}(\lambda_{1}+\lambda_{12}) \exp(-(\lambda_{1}+\lambda_{12})x_1-\lambda_{2}x_2) & x_1>x_2\\
\lambda_{12}\exp(-\Lambda x_1) & x_1=x_2
\end{cases}$\\
&$\Lambda=$&$\lambda_1+\lambda_2+\lambda_{12}$\\
\hline
 Survival function &$S(x_1,x_2)=$& $\exp(-\lambda_1 x_1 - \lambda_2 x_2 - \lambda_{12}\max(x_1,x_2))$  \\
\hline
Expectations &$E[X_1] =$&$\frac{1}{\lambda_1 + \lambda_{12}}$\\
&$E[X_2] =$&$\frac{1}{\lambda_2 + \lambda_{12}}$ \\
&$E[X_{(1)}] =$&$\frac{1}{\Lambda}$\\
&$E[X_{(1)}|X_{1}=X_{2}] =$&$\frac{1}{\Lambda}$\\
&$E[X_{(2)}]=$&$\frac{1}{\Lambda}+\frac{\lambda_1}{\Lambda (\lambda_2+\lambda_{12})}+\frac{\lambda_2}{\Lambda (\lambda_1+\lambda_{12})}$\\
&$E[X_{(2)}|X_{1}\neq X_{2}]= $&$\frac{1}{\Lambda}+\frac{1}{\lambda_{1}+\lambda_{2}}(\frac{\lambda_{1}}{\lambda_{2}+\lambda_{12}}+\frac{\lambda_{2}}{\lambda_{1}+\lambda_{12}})$ \\
&$E[TBE]=$&$0.5(\frac{\lambda_{2}}{\Lambda^{2}}+\frac{\lambda_{2}}{\Lambda(\lambda_{1}+\lambda_{12})}+\frac{\lambda_{1}}{\Lambda^{2}}+\frac{\lambda_{1}}{\Lambda(\lambda_{2}+\lambda_{12})})+\frac{\lambda_{12}}{\Lambda^{2}}$ \\
\hline
Probabilities&$P[X< Y] =$ &$\frac{\lambda_1}{\Lambda}$\\
&$P[X>Y] =$ &$\frac{\lambda_2}{\Lambda}$\\
&$P[X=Y] =$& $\frac{\lambda_{12}}{\Lambda}$\\
\hline \hline
\end{tabular}
\label{tab:mobe}
\end{table}

\subsection*{B.3 Marshall Olkin Bivariate Weibull distribution}
In the MOBW model, random shocks affect the system and they are modelled as a non-homogeneous Poisson process. The MOBW model was developed by \cite{marshall1967multivariate} and its survival function with parameters $ \lambda_1,\lambda_2, \lambda_{12}, \eta $ is 
\begin{equation*}
	S(x_1,x_2) = \exp(-\lambda_1 x_1^{\eta} - \lambda_2 x_2^{\eta} - \lambda_{12} \max(x_1,x_2)^{\eta}), \,\,\, x_1, x_2 > 0
\end{equation*}

For parameter estimation of the MOBW model parameters, one can derive the maximum likelihood estimators by solving the following equations:$\hat{\lambda}_1 (\eta) = \frac{n_1} {\sum_{i=1}^{n}(r_i + 1)y_{i}^{\eta}}$ , $ \hat{\lambda}_2 (\eta) =  \frac{n_2} {\sum_{i=1}^{n}(r_i + 1)y_{i}^{\eta}}$, $ \hat{\lambda}_{12} (\eta) = \frac{n_12} {\sum_{i=1}^{n}(r_i + 1)y_{i}^{\eta}}$ and $ \hat{\eta} = h(\eta)$. For more details on deriving these estimates and more explanation of the symbols, the reader is referred to \cite{feizjavadian2015analysis}. For obtaining estimates by EM algorithm the reader is referred to \cite{kundu2009estimating}.

In order to simulate data according to the MOBW model, one first obtains $P, Q$ and $R$ as independent Weibull distributed variables with common shape parameter $\eta$ and scale parameters ${\lambda_1}^{-1/\eta}$, ${\lambda_2}^{-1/\eta}$ and ${\lambda_{12}}^{-1/\eta}$, respectively. 
Next, compute $X_1=min(P,R)$ and $X_2=min(Q,R)$. Table \ref{tab:mobw} provides more details on MOBW distributed event data.
\begin{table}[h]
\caption{Characteristics of the MOBW distribution}
\centering
\footnotesize
\begin{tabular}{lrl}
\hline \hline
pdf &$f(x_1,x_2) =$& 
$\begin{cases}
\eta^{2}\lambda_{1}(\lambda_{2}+\lambda_{12})x_1^{\eta-1}x_2^{\eta-1}\exp(-\lambda_{1}x_1^{\eta}-(\lambda_{2}+\lambda_{12})x_2^{\eta}) & x_1<x_2\\
\eta^{2}\lambda_{2}(\lambda_{1}+\lambda_{12})x_1^{\eta-1}x_2^{\eta-1}\exp(-(\lambda_{1}+\lambda_{12})x_1^{\eta}-\lambda_{2}x_2^{\eta}) & x_1>x_2\\
\eta\lambda_{12}x_1^{\eta-1}\exp(-(\lambda_{1}+\lambda_{2}+\lambda_{12})x_1^{\eta}) & x_1=x_2
\end{cases}$\\
&$\Lambda=$&$\lambda_1+\lambda_2+\lambda_{12}$\\
\hline
 Survival function &$S(x_1,x_2)=$& $ \exp(-\lambda_1 x_1^{\eta} - \lambda_2 x_2^{\eta} - \lambda_{12} \max(x_1,x_2)^{\eta}), \,\,\, x_1, x_2 > 0$  \\
\hline
Expectations &$E[X_1] =$&$\Gamma(1+\frac{1}{\eta})\frac{1}{(\lambda_{1}+\lambda_{12})^{1/\eta}}$\\
&$E[X_2] =$&$\Gamma(1+\frac{1}{\eta})\frac{1}{(\lambda_{2}+\lambda_{12})^{1/\eta}}$\\
&$E[X_{(1)}|X_{1}=X_{2}] =$&$\Gamma(1+\frac{1}{\eta})\frac{1}{\Lambda^{\frac{1}{\eta}}}$\\
&$E[X_{(2)}] =$&$ \Gamma(1+\frac{1}{\eta})(\frac{1}{(\lambda_{2}+\lambda_{12})^{1/\eta}}+\frac{1}{\Lambda^{1/\eta}})$\\
&$E[X_{(2)}|X_{1}\neq X_{2}] =$&$ \Gamma(1+\frac{1}{\eta})(\frac{\Lambda}{(\lambda_{2}+\lambda_{12})^{1/\eta}(\lambda_{1}+\lambda_{2})}+\frac{\Lambda}{(\lambda_{1}+\lambda_{12})^{1/\eta}(\lambda_{1}+\lambda_{2})}-\frac{\Lambda+\lambda_{12}}{\Lambda^{\frac{1}{\eta}}(\lambda_{1}+\lambda_{2})})$\\
&$E[TBE]=$&$0.5\Gamma(1+\frac{1}{\eta})(\frac{1}{(\lambda_{2}+\lambda_{12})^{1/\eta}}-\frac{\lambda_{2}+\lambda_{12}}{\Lambda^{1+1/\eta}}+\frac{1}{(\lambda_{1}+\lambda_{12})^{1/\eta}}-\frac{\lambda_{1}+\lambda_{12}}{\Lambda^{1+1/\eta}}+2\frac{\lambda_{12}}{\Lambda^{1+1/\eta}})$ \\
\hline
Probabilities&$P[X< Y] =$ &$\frac{\lambda_1}{\Lambda}$\\
&$P[X>Y] =$ &$\frac{\lambda_2}{\Lambda}$\\
&$P[X=Y] =$& $\frac{\lambda_{12}}{\Lambda}$\\
\hline \hline
\end{tabular}
\label{tab:mobw}
\end{table}

\section*{Appendix C: The EM algorithm for MOBW estimates}\label{Sec:GBE_EMalgo}
The MOBW distribution is observed as a shock model where the shocks are occurring as a non-homogeneous Poisson process. \cite{kundu2009estimating} provided an EM algorithm to obtain maximum likelihood estimates of the MOBW parameters. Their EM algorithm works when the data belong to all of the following sets;  
\begin{equation*}
	I_0=[i;X_{1i}=X_{2i}],  I_1=[i;X_{1i}<X_{2i}], I_2=[i;X_{1i}>X_{2i}],
\end{equation*}
where $i$ denotes subject. It is clearly seen from Figure~\ref{AIDS_data_plot} that the  AIDS dataset belongs to set $I_1$ only. To obtain the parameters of the AIDS dataset, we have redefined the EM algorithm as follows.

The log-likelihood function for case $I_1$ where $n= \vert{I_1} \vert$ can be written as,
\begin{equation}\label{log-lik_MOBW}
	\begin{aligned}
	l(\eta,\lambda_1,\lambda_2,\lambda_{12})= 
	&n\;ln(\eta\lambda_1)+n\; ln(\eta(\lambda_2+\lambda_{12}))+(\eta-1) \left[\sum_{i=1}^{n} ln(X_{1i})+\sum_{i=1}^{n} ln(X_{2i})\right]	\\
	&-\lambda_1 \sum_{i=1}^{n} X_{1i}^\eta -(\lambda_2+\lambda_{12}) \sum_{i=1}^{n} X_{2i}^\eta.
\end{aligned}
\end{equation}
To implement the EM algorithm, we obtain the \emph{E} step similarly as in \citep{kundu2009estimating} for which the pseudo-log-likelihood function derived by the log-likelihood function given in Equation~\ref{log-lik_MOBW} is defined as follows,
\begin{equation}\label{ps-log-lik_MOBW}
	\begin{aligned}
		l_{pseudo}(\eta,\lambda_1,\lambda_2,\lambda_{12})=
		&2n\;ln(\eta)+(\eta-1) \left[\sum_{i=1}^{n} ln(X_{1i})+\sum_{i=1}^{n} ln(X_{2i})\right]-\lambda_{12} \sum_{i=1}^{n} X_{2i}^\eta\\
		& + n\; \frac{\lambda_{12}}{(\lambda_1+\lambda_2)}\;ln(\lambda_{12})
		-\lambda_1 \sum_{i=1}^{n} X_{1i}^\eta+n\;ln(\lambda_1)-\lambda_2 \sum_{i=1}^{n} X_{2i}^\eta\\
		&+n\;\frac{\lambda_2}{(\lambda_{12}+\lambda_2)}\;ln(\lambda_2).
	\end{aligned}
\end{equation}
 For more details on pseudo-log-likelihood function see \cite{dinse1982nonparametric} and \cite{kundu2009estimating}. Further,  \emph{M} step involves maximizing the pseudo-log-likelihood function given in Equation~\ref{ps-log-lik_MOBW} with respect to $\eta,\lambda_1$, $\lambda_2$ and $\lambda_{12}$. It is noted that for the fixed $\eta$, the maximizing of Equation~\ref{ps-log-lik_MOBW} with respect to $\lambda_1,\lambda_2$ and $\lambda_{12}$ can be obtained as,
\begin{equation*}
	\begin{aligned}
		\hat{\lambda}_0(\eta)=\frac{n \lambda_{12}/(\lambda_1+\lambda_2)}{\sum_{i=1}^{n} X_{2i}^\eta}
	\end{aligned}
\end{equation*}

\begin{equation*}
	\begin{aligned}
		\hat{\lambda}_1(\eta)=\frac{n}{\sum_{i=1}^{n} X_{1i}^\eta}
	\end{aligned}
\end{equation*}

\begin{equation*}
	\begin{aligned}
		\hat{\lambda}_2(\eta)=\frac{n\lambda_2/(\lambda_{12}+\lambda_2)}{\sum_{i=1}^{n} X_{2i}^\eta}.
	\end{aligned}
\end{equation*}
Furthermore, the maximizing of Equation~\ref{ps-log-lik_MOBW} with respect to $\eta$ can be obtained by solving a fixed point type equation
\begin{equation*}
	\begin{scriptsize}
	\begin{aligned}
	g(\eta)=\frac{2n}{\hat{\lambda}_0(\eta) \sum_{i=1}^{n} X_{2i}^\eta\; ln(X_{2i})+\hat{\lambda}_1(\eta) \sum_{i=1}^{n} X_{1i}^\eta\; ln(X_{1i})+\hat{\lambda}_2(\eta) \sum_{i=1}^{n} X_{2i}^\eta\; ln(X_{2i})-\left[\sum_{i=1}^{n} ln(X_{1i})+\sum_{i=1}^{n} ln(X_{2i})\right]}.
	\end{aligned}
	\end{scriptsize}
\end{equation*}
The steps for implementation of the redefined EM algorithm were the same with the steps of EM algorithm proposed by \cite{kundu2009estimating}.

\end{document}